\documentclass[11pt]{article}

\usepackage{amsmath}
\usepackage{fullpage}
\usepackage{ifthen}
\usepackage{xspace}
\usepackage{definitions}
\usepackage{bm}

\newcommand{\NUMCAND}{\ensuremath{n}\xspace}

\newcommand{\COST}{\ensuremath{\bm{c}}\xspace}

\newcommand{\Cost}[1]{\ensuremath{c_{#1}}\xspace}
\newcommand{\CostP}[1]{\ensuremath{c'_{#1}}\xspace}

\newcommand{\CPROB}{\ensuremath{\bm{p}}\xspace}

\newcommand{\CProb}[1]{\ensuremath{p_{#1}}\xspace}
\newcommand{\CProbP}[1]{\ensuremath{p'_{#1}}\xspace}
\newcommand{\VPROB}{\ensuremath{\bm{q}}\xspace}
\newcommand{\VProb}[1]{\ensuremath{q_{#1}}\xspace}

\newcommand{\DIST}{\ensuremath{\mathcal{D}}\xspace}
\newcommand{\Dist}[2]{\ensuremath{d_{#1,#2}}\xspace}
\newcommand{\WINNER}{\ensuremath{w}\xspace}
\newcommand{\Winner}[2]{\ensuremath{\WINNER(#1,#2)}\xspace}
\newcommand{\OPT}{\ensuremath{o}\xspace}
\newcommand{\Opt}[2]{\ensuremath{\OPT(#1,#2)}\xspace}
\newcommand{\Social}[3]{\ensuremath{C\left(#1,#2,#3\right)}\xspace}
\newcommand{\CSOC}{\ensuremath{\hat{C}}}
\newcommand{\CSoc}[3][]{\ensuremath{\ifthenelse{\equal{#1}{}}{\CSOC(#2,#3,\alpha)}{\CSOC_{#2,#3,\alpha}(#1)}}\xspace}
\newcommand{\Last}[1]{\ensuremath{\ell_{#1}}\xspace}

\newcommand{\Ratio}[2]{\ensuremath{\ifthenelse{\equal{#2}{}}{r_{#1}}{r_{#1,#2}}}\xspace}
\newcommand{\RatioP}[2]{\ensuremath{\ifthenelse{\equal{#2}{}}{r'_{#1}}{r'_{#1,#2}}}\xspace}
\newcommand{\Radius}[1]{\ensuremath{\rho_{#1}}\xspace}
\usepackage{graphicx}
\usepackage{hyperref}
\usepackage{pstricks,pst-node}

\usepackage{color}
\usepackage{epigraph}

\newtheorem{example}[lemma]{Example}

\renewcommand{\hat}{\widehat}

\newcommand\eps{\epsilon}

\newcommand{\expect}[2]{\Expect[#1]{#2}}

\begin{document}

\title{Of the People: Voting Is More Effective \\ with Representative Candidates}
\author{Yu Cheng \qquad Shaddin Dughmi \qquad David Kempe \\
Dept.~of Computer Science, University of Southern California}
\date{}

\maketitle

\begin{abstract}
In light of the classic impossibility results of Arrow and
Gibbard and Satterthwaite regarding voting with ordinal rules,
there has been recent interest in characterizing how well common
voting rules approximate the social optimum.
In order to quantify the quality of approximation, it is natural to
consider the candidates and voters as embedded within a common metric
space, and to ask how much further the chosen candidate is from the
population as compared to the socially optimal one.
We use this metric preference model to explore a fundamental and
timely question: 
does the social welfare of a population improve when candidates are
representative of the population?
If so, then by how much, and how does the answer depend on the
complexity of the metric space?

We restrict attention to the most fundamental and common social choice
setting: a population of voters, two candidates,
and a majority rule election.
When candidates are not representative of the population, it is known
that the candidate selected by the majority rule can be thrice as far
from the population as the socially optimal one; 
this holds even when the underlying metric is a line.
We examine how this ratio improves when candidates are drawn
independently from the population of voters.
Our results are two-fold: When the metric is a line, the ratio
improves from $3$ to $(4-2\sqrt{2}) \approx 1.1716$; this bound is tight. 
When the metric is arbitrary, we show a lower bound of 1.5 and a
constant upper bound strictly better than 2 on the distortion
of majority rule.

The aforementioned positive results depend in part on the assumption
that the two candidates are independently and identically distributed.
However, we show that i.i.d.~candidates do not suffice for our upper bounds: 
if the population of candidates can be different from that of voters, 
an upper bound of 2 on the distortion is tight for both general metric
spaces and the line.
Thus, we show a constant gap between representative and
non-representative candidates in both cases.
The exact size of this gap in general metric spaces is a natural open question.
\end{abstract}


\section{Introduction}

\epigraph{\it ``[...] and that government of the people, by the people, for the
people, shall not perish from the earth.''}{--- \textup{Abraham Lincoln}}


Abraham Lincoln's Gettysburg Address culminated with the oft-quoted
words above.
This single sentence gives a remarkably succinct summary of the role of
a country's populace in a participatory democracy, identifying three
distinct facets:
(1) The government should be \emph{of} the people: 
the members of the government should be drawn from --- and by
inference representative of --- the country's populace.
(2) The government should be \emph{by} the people: 
decisions should be made by the populace.
(3) The government should be \emph{for} the people: 
its objective should be to serve the interests of the populace.
In Lincoln's words, the central question we study here is the following:

\begin{quote} 
If a government by the people is to be for the people, 
how important is it that it also be \emph{of} the people?
\end{quote}

In quantifying this question, we observe that there is a surprisingly
clean mapping of Lincoln's vision onto central concepts of social
choice theory:

\begin{enumerate}
\item Who is the government of? 
  Who are the candidates (people or ideas) to be aggregated?
\item Who is the government by? 
  What are the social choice rules used for aggregation?
\item Who is the government for? 
  What objective function is to be optimized?
\end{enumerate}

While the exact social choice rules to be used have been a topic of
vigorous debate for several centuries
\cite{borda:elections,condorcet:essay,arrow:social-choice,BCULP:social-choice},
the broad class they are drawn from is generally agreed upon:
voters provide an ordinal ranking of (a subset of) the candidates, and
these rankings are then aggregated to produce either a single winner
or a consensus ranking of all (or some) candidates.
Social choice is limited by the severe impossibility results
of Arrow~\cite{arrow:social-choice} and Gibbard and Satterthwaite
\cite{gibbard:manipulation,satterthwaite:voting}, 
establishing that even very simple
combinations of desired axioms are in general unachievable.
These impossibility results in turn have resulted in a fruitful line
of work exploring restrictions on individuals' preference orders for
circumventing the impossibility of social choice.

One of the avenues toward circumventing the impossibility results
simultaneously doubles as a framework for addressing the third
question: 
What objective function is to be optimized by the social choice rule?
The key modeling assumption is that all candidates (ideas or people)
and voters are embedded in a metric space: small distances model
high agreement, while large distances correspond to disagreement 
\cite{black:rationale,downs:democracy,black:committees-elections,moulin:single-peaked,merrill:grofman,barbera:gul:stacchetti,richards:richards:mckay,barbera:social-choice}.
The metric induces a preference order over
  candidates for each voter:
she simply ranks candidates by distance from herself.
When the metric space is specifically the line, we obtain the
  well-known and much studied special case of single-peaked
  preferences \cite{black:rationale,moulin:single-peaked}.
Embedding voters and candidates in a metric space has
historically served two purposes: 
(1) Restricting the metric space --- for example, by limiting
its dimension --- defines a restricted class of ordinal preference
profiles, and might help circumvent
the classic impossibility results of social choice.
(2) The distances naturally provide an objective function: 
the \emph{best} alternative is the one that is closest to the voters
on average.  Even when the metric space is unrestricted, replacing the
hard axioms of social choice theory with this objective function can
``circumvent'' impossibility results through approximation
\cite{procaccia:approximation:gibbard}, and permits comparing different
social rules by quantifying their worst-case performance.  

While distances yield cardinal preferences and a social objective
function, it is arguably unrealistic to expect individuals to
articulate distances accurately. 
It is consequently unsurprising that common and well-established
voting rules typically restrict voters to providing ordinal
information, such as rankings or a single vote. 
Therefore, we view the metric space as implicit, and a social
choice function as optimizing the associated cardinal objective
function \emph{using only ordinal information}. 


This viewpoint was recently crisply expressed in a sequence of works
originating with Anshelevich et
al.~\cite{anshelevich:bhardwaj:postl,anshelevich:postl:randomized,anshelevich:sekar:blind,anshelevich:ordinal,goel:krishnaswamy:munagala}.
In particular, Anshelevich et al.~\cite{anshelevich:bhardwaj:postl}
examine many of the most widely used election
voting rules, guided by the question: ``How much worse is the outcome
of voting than would be the omniscient choice of the best available candidate?''
They showed remarkable separations: while some voting rules guarantee
a distortion of no more than a constant factor, others are off by a
factor that increases linearly in the number of candidates or --- even
worse --- voters.
The simplest, and in some sense canonical, example of such
distortion is captured as follows:

\begin{example} \label{ex:line-three}
A population consists of voters of whom just below half
lean solidly left (at position $-1$), 
while just over half are just to the right of center (at position
$\epsilon > 0$). 
The population conducts an election between a solidly
left-wing (position $-1$) and a solidly right-wing (position 1)
candidate. 
\begin{figure}[h]
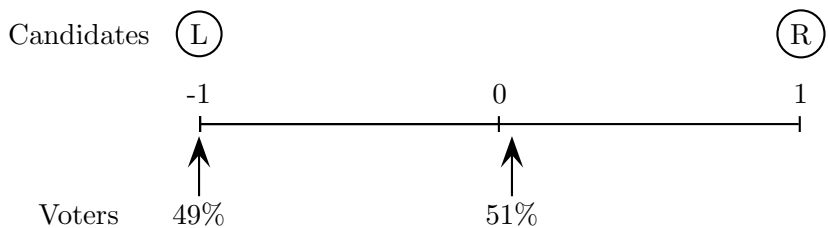

\centering
\psset{unit=0.8cm,arrowsize=0.2 5}
\pspicture(-7,-2)(6,2)
\psline{|-|}(-5,0)(5,0)
\psline{|-|}(-5,0)(0,0)

\rput(-5,0.5){-1}
\rput(5,0.5){1}
\rput(0,0.5){0}

\rput(-5,1.5){\pscirclebox{L}}
\rput(5,1.5){\pscirclebox{R}}

\psline{->}(-5,-1.2)(-5,-0.2)
\rput(-5,-1.5){49\%}
\psline{->}(0.2,-1.2)(0.2,-0.2)
\rput(0.2,-1.5){51\%}

\rput(-7,-1.5){Voters}
\rput(-7,1.5){Candidates}
\endpspicture

\caption{The winning candidate could have thrice the social cost of the other candidate.}
\label{fig:thrice-example}
\end{figure}

Because the centrist voters express their (slight) preference
for the right-wing candidate, he is elected by a small
majority. 
However, the average distance from the population to the
right-wing candidate (1.5) is thrice that to the left-wing
candidate (0.5), meaning that the majority vote led to a loss of a
factor three in the utility.
\end{example}

We follow prior nomenclature in this domain~
\cite{procaccia:rosenschein:distortion,BCHLPS:utilitarian:distortion,caragiannis:procaccia:voting,boutilier:rosenschein:incomplete,anshelevich:bhardwaj:postl}
and term this utility loss the \emph{distortion}.
In examining Example~\ref{ex:line-three} more closely, 
we identify a likely culprit for the high distortion: 
the right-wing candidate was not representative of the population
--- he was not \emph{of} the people. 
Had we drawn two candidates from the population, the winner would
in fact always be the socially optimal choice in this case.
If we wanted to create the possibility of recreating the above example, 
we would need to move some fraction $\delta$ of the
population to the right wing.
If $\delta$ were large, then the election of a right-wing candidate
would not be nearly as bad according to the objective function;  
conversely, if $\delta$ were small, then it would be unlikely that a
right-wing candidate would run, so most of the time, the social choice
rule would select an optimal candidate.
Thus, intuitively, when candidates are drawn from the population, 
we would expect the distortion in the social cost to be better than
when they are not.
The goal of this article is to investigate to what extent this
intuition holds.

\subsection*{The Model}
Formally, we assume that the candidates and voters are jointly located
in a (finite) metric space; the distance between $i$ and $j$ is
denoted by \Dist{i}{j}.
The candidates' locations are given by a probability distribution \CPROB, 
while the voters' location distribution is denoted by \VPROB.
In order to isolate the driving question and side-step issues of
specific voting rules, we focus on the simplest social choice
scenario: two candidates $i,i'$ are drawn i.i.d.~from \CPROB,
and a simple majority vote determines the winner between them.
Voter $j$ votes for the one of $i,i'$ who is closer\footnote{%
\label{ftn:tie-breaking}
  Throughout, we will assume when convenient that the metric
    and distribution are in general position. Specifically, there are no
    ties in any voter's preference order,
   and there are no ties in any election outcome.
   Ties could in principle be dealt with using suitable tie breaking
   rules, but the slight gain in generality would not be worth the overhead.}
to $j$.
The \emph{social cost} of candidate $i$ is 
$\Cost{i} = \sum_j \VProb{j} \Dist{i}{j}$.
With $\Winner{i}{i'}$ denoting the winner of the election
and $\Opt{i}{i'}$ the socially optimal candidate, 
the expected distortion of voting is
$\sum_{i,i'} \CProb{i} \CProb{i'} \frac{\Cost{\Winner{i}{i'}}}{\Cost{\Opt{i}{i'}}}$.
Our goal is then to understand whether and by how much the distortion
decreases when candidates are of the people (when $\CPROB = \VPROB$).

\subsection*{Our Results}
We begin our investigation with arguably the simplest metric space,
which nonetheless is frequently used to describe the political
spectrum of countries: the line.
As we saw in Example~\ref{ex:line-three}, even for the line, voting
between two arbitrary candidates can lead to a distortion of 3.
Our first main result (proved in Section~\ref{sec:line}) is that when
two candidates are drawn i.i.d.~from 
$\CPROB = \VPROB$, the expected distortion is at most
$4-2\sqrt{2} \approx 1.1716$, and this bound is tight.
The lower-bound example is in fact of the type discussed after
Example~\ref{ex:line-three}, obtained by moving a suitable
population mass $\delta$ from location $\epsilon$ to
location 1.
The more difficult part of the proof is the upper bound, and in
particular, the proof that the worst-case distribution of
voters/candidates always has support size 3.
The proof proceeds by showing that for larger support sizes, there is
always a sequence of alterations that gradually shifts the population to
fewer locations, without lowering the distortion.

\smallskip

Next, we turn our attention to general metric spaces.
For arbitrary metric spaces, the distortion of voting can be larger.
In Section~\ref{sec:metric}, we analyze a simple example:
just under half the population is located at one point $i$, while the rest
of the population is spread out evenly over $\NUMCAND \gg 1$ locations
that are at distances just below 1 from each other and at distance 1
from $i$. 
As $\NUMCAND \to \infty$, we show that the expected distortion
converges to $\frac{3}{2}$.
The upper bound we establish in Section~\ref{sec:metric} does not
match this lower bound: we show that for every metric and every
\CPROB, the expected distortion is at most $2-\frac{1}{652}$.
We conjecture that the bound of $\frac{3}{2}$ is in fact tight ---
proving or disproving this conjecture is a natural direction for
future work, discussed in Section~\ref{sec:conclusion}.

The significance of our upper bounds on distortion (for the line and
for general metric spaces)
arises from the contrast to the corresponding bounds when
$\VPROB \neq \CPROB$.
In revisiting the improved distortion results we prove, 
we notice two potential driving factors: 
(1) The two candidates are independently and identically distributed. 
(2) The distributions of candidates and voters are the same.
One may wonder whether the innocuous-looking assumption of
i.i.d.~candidates alone could be responsible for the lower
distortion, without requiring that candidates be of the people.
In Section~\ref{sec:different}, we rule out this possibility by
establishing a (tight) bound of 2 on the distortion of voting when
candidates are drawn i.i.d.~from $\CPROB \neq \VPROB$, 
both in general metrics and on the line. 
The (small, but constant) gap between the distortions of 
$2-\frac{1}{652}$ and 2 in general metric spaces, 
and the significant gap between the distortions of $4-2\sqrt{2}
\approx 1.1716$ and $2$ on the line, 
show that government by the people is more efficient
when it is also \emph{of the people}.
The exact size of the gap between the two distortions in general
metric spaces is a natural open question.

\subsection*{Related Work}
There has been a lot of interest recently in circumventing
the impossibility results of voting and social choice by
approximation; see, e.g.,
\cite{procaccia:rosenschein:distortion,procaccia:approximation:gibbard,caragiannis:procaccia:voting}
and \cite{boutilier:rosenschein:incomplete} for a recent survey.
Of particular interest is the recent direction in which the voters'
objective functions are derived from proximity in a metric space
\cite{anshelevich:bhardwaj:postl,anshelevich:postl:randomized,anshelevich:sekar:blind,anshelevich:ordinal,goel:krishnaswamy:munagala,feldman:fiat:golomb}.
One of the important issues is providing incentives for truthful
revelation of preferences (e.g., \cite{feldman:fiat:golomb});
in this paper, we side-step this issue by considering only
elections between two candidates at a time.

Our work is most directly inspired by the recent work of 
Anshelevich et al.~\cite{anshelevich:bhardwaj:postl,anshelevich:postl:randomized},
which analyzes the distortion of ordinal voting rules when evaluated
for metric preferences.
Our work departs from \cite{anshelevich:bhardwaj:postl,anshelevich:postl:randomized}
in assuming that the candidates themselves are drawn i.i.d.~from underlying
distributions, and in particular in analyzing the case when the
distribution of the candidates is equal to that of the voters.

Anshelevich and Postl~\cite{anshelevich:postl:randomized} consider a
condition of instances that also aims to capture that candidates are
in some sense ``representative'' of the voting population.
Specifically, they define a notion of \emph{decisiveness} as follows:
Let $i$ be a voter, and $j_i, j'_i$ her two closest candidates, with
$\Dist{i}{j_i} \leq \Dist{i}{j'_i}$.
An instance is $\alpha$-decisive (for $\alpha \leq 1$)
if $\Dist{i}{j_i} \leq \alpha \Dist{i}{j'_i}$ for all $i$; 
in other words, when $\alpha \ll 1$, 
every voter has a strongly preferred candidate.
Naturally, the decisiveness condition is applicable only in elections
in which the number of candidates is large or the space of voters is
highly clustered.
In our work, by considering candidates drawn from the voter
distribution, we avoid such assumptions.


\section{Preliminaries}
\label{sec:preliminaries}

The candidates and voters are embedded in 
  a finite metric space $\DIST = (\Dist{i}{j})_{i,j}$
  with \emph{points (locations)} $i=1,\ldots,\NUMCAND$.
Depending on the context, we will refer to $i$ as a point, candidate, or voter. 
The probability for a candidate to be drawn from point $i$ is
\CProb{i}; we write $\CPROB = (\CProb{i})_i$.
The fraction of voters at $i$ is \VProb{i}, 
summarized as $\VPROB = (\VProb{i})_i$.
For a subset of points $A$, 
we write $\CProb{A} = \sum_{i\in A} \CProb{i}$
to denote the total probability mass in $A$, 
and similarly for \VProb{A}.
The \emph{social cost} of a candidate $i$ is his
average distance to all voters:
\begin{align}
\Cost{i} & = \sum_{j} \VProb{j} \cdot \Dist{i}{j}.
\label{eqn:cost-def}
\end{align}

When candidates $i$ and $i'$ are competing, 
each voter $j$ votes for the candidate that is closer\footnote{Recall
  the discussion of tie breaking in Footnote~\ref{ftn:tie-breaking}.} to her,
i.e., for $\argmin_{i,i'} (d(j,i), d(j,i'))$.
The \emph{winner} is the candidate who gets more votes:
$i$ wins iff $\sum_{j: \Dist{i}{j} \leq \Dist{i'}{j}} \VProb{j} \geq \half$.
For two candidates $i,i'$, 
let \Winner{i}{i'} denote the winner as just described, 
and let $\Opt{i}{i'} = \argmin_{j \in \SET{i,i'}} \Cost{j}$ be the
candidate of lower social cost.
The \emph{distortion} of an election between two candidates $(i, i')$
is defined as
\begin{align*}
\Ratio{i}{i'} = \frac{\Cost{\Winner{i}{i'}}}{\Cost{\Opt{i}{i'}}}.
\end{align*}

We are interested in the (expected) \emph{distortion} of the instance
$(\DIST, \CPROB, \VPROB)$, defined as
the expected distortion of an election between two candidates drawn
i.i.d.~from the candidate distribution \CPROB:
\begin{align}
\Social{\DIST}{\CPROB}{\VPROB}
& = \Expect[i,i' \sim \CPROB]{\Ratio{i}{i'}}
  = \Expect[i,i' \sim \CPROB]{\frac{\Cost{\Winner{i}{i'}}}{\Cost{\Opt{i}{i'}}}}
\; = \; 2 \sum_{i < i'} \CProb{i} \CProb{i'} \cdot \frac{\Cost{\Winner{i}{i'}}}{\Cost{\Opt{i}{i'}}}
       + \sum_{i} \CProb{i}^2 \cdot 1.
\label{eqn:social-cost-def}
\end{align}

In particular, our goal is to analyze the worst-case
distortion when the candidates are representative and when they are
not, that is, we want to find the gap between
\begin{equation*}
\max_{\DIST, \CPROB, \VPROB} \Social{\DIST}{\CPROB}{\VPROB} \quad \text{and} \quad
  \max_{\DIST, \CPROB} \Social{\DIST}{\CPROB}{\CPROB}.
\end{equation*}


\newcommand{\LINE}{\ensuremath{\mathcal{L}}\xspace}

\section{Identical Distributions on the Line}
\label{sec:line}

We begin with the simplest setting: the underlying metric space is the line,
and two candidates are drawn independently from the population of voters ($\CPROB = \VPROB$).
We first show a family of examples (a variant of Example~\ref{ex:line-three})
for which the expected distortion gets arbitrarily close to $4-2\sqrt{2} \approx 1.1716$.

\begin{example} \label{ex:line-iid}
The metric space is the line, denoted by \LINE. 
There are $\CProb{1}=\frac{1}{2}-\eps$ voters at location $x_1 = -1$, 
$\CProb{2}=1-\frac{1}{\sqrt{2}}$ voters at $x_2 = \eps$, and 
$\CProb{3}=\frac{1}{\sqrt{2}}-\frac{1}{2}+\eps$ voters at $x_3 = 1$.
This example is obtained  from Example \ref{ex:line-three}
by moving a suitable fraction of voters from location $x_2 = \eps$ to $x_3 = 1$,
carefully trading off between two factors:
(1) decreasing the pairwise distortion between the candidates at $-1$ and $1$,
but (2) increasing the chance of a such an election happening.
\begin{figure}[h]
\centering
\psset{unit=0.8cm,arrowsize=0.2 5}
\pspicture(-9,-2)(6,1)
\psline{|-|}(-5,0)(5,0)
\psline{|-|}(-5,0)(0,0)

\rput(-5,0.5){-1}
\rput(5,0.5){1}
\rput(0,0.5){0}

\psline{->}(-5,-1.2)(-5,-0.2)
\rput(-5,-1.5){49.99\%}
\psline{->}(0.2,-1.2)(0.2,-0.2)
\rput(0.2,-1.5){29.29\ldots\%}
\psline{->}(5,-1.2)(5,-0.2)
\rput(5,-1.5){20.71\ldots\%}

\rput(-8.5,-1.5){Voters/Candidates}
\endpspicture

\caption{The worst case instance on the line with $\Social{\LINE}{\CPROB}{\CPROB} = 4 - 2\sqrt{2}$.}
\label{fig:example3}
\end{figure}

Because the voters at $x_2 = \eps$ are slightly closer to 1 than to -1,
a candidate drawn from $x_3 = 1$ will win against a candidate
drawn from $x_1 = -1$.
The costs of the two candidates are
\begin{align*}
  \Cost{1} &= \CProb{2}\Dist{1}{2} + \CProb{3}\Dist{1}{3}
    = \CProb{2} + 2\CProb{3} + O(\eps) = \frac{1}{\sqrt{2}} + O(\eps), \\
  \Cost{3} &= \CProb{1}\Dist{1}{3} + \CProb{2}\Dist{2}{3}
    = 2\CProb{1} + \CProb{2} - O(\eps) = 2 - \frac{1}{\sqrt{2}} - O(\eps).
\end{align*}
Because the candidates are drawn independently from \CPROB,
the election between $x_1$ and $x_3$ happens with probability $2 \CProb{1}\CProb{3}$.
In all other cases (when a candidate from $x_2$ runs against one from
$x_1$ or $x_3$, or both candidates are from the same location), 
the voters elect the socially better candidate.
Therefore, the expected distortion is
\begin{align*}
\Social{\LINE}{\CPROB}{\CPROB} 
& = (1 - 2\CProb{1}\CProb{3}) \cdot 1 + (2\CProb{1}\CProb{3}) \cdot \frac{\Cost{3}}{\Cost{1}} 
\; = \; 4-2\sqrt{2} - O(\eps).
\end{align*}
\end{example}

Our first main result is that Example~\ref{ex:line-iid} gives the worst distortion on the line.

\begin{theorem} \label{thm:lineWorst}
For any distribution \CPROB,
we have $\Social{\LINE}{\CPROB}{\CPROB} \le 4 - 2\sqrt{2}$.
\end{theorem}

We will prove Theorem~\ref{thm:lineWorst} in Section~\ref{sec:line-ub}.
In preparation, in Section~\ref{sec:line-order},
we first provide some structural characterization results about the
voting behavior and social cost on the line.
  
\subsection{Characterizing the Structure of Voting on the Line}
\label{sec:line-order}
Given a distribution on the line with support size $n$,
we label the support points as $1, \ldots, n$ from left to right.
Let $m$ be the index of the median\footnote{Recall that we assume the
  instance to be in general position, which implies uniqueness of the median.}, 
and let $L = \SET{1,\ldots,m-1}$ and $R = \SET{m+1,\ldots,n}$
denote the locations to the left and to the right of the median, respectively.
By the definition of the median,
$\CProb{L} < \half < \CProb{L} + \CProb{m}$ 
and $\CProb{R} < \half < \CProb{m} + \CProb{R}$.

\begin{lemma} \label{lem:voteOrder}
If two candidates $(x,y)$ are drawn, the one closer to $m$ wins the election.
\end{lemma}

\begin{proof}
Without loss of generality, we assume that $d_{x,m} < d_{y,m}$ and $x \in L \cup \{m\}$;
that is, $x$ lies to the left of the median, or $x$ is the median.
There are two cases depending on whether $y$ is also to the left of $m$.
\begin{enumerate}
\item If $y \in L$,
  then all voters to the right of the median as well as the median are
  going to vote for $x$, so $x$ gets a $\CProb{m} + \CProb{R} > \half$
  fraction of the votes.
\item If $y \in R$,
  then all voters in $L$ as well as $m$ are going to vote for $x$,
  so $x$ gets a $\CProb{L} + \CProb{m} > \half$ fraction of the votes.
\end{enumerate}
In either case, $x$ gets more than half of the votes and wins the election.
\end{proof}

The next lemma characterizes the social cost ordering on the line.
\begin{lemma} \label{lem:socialOrder}
If $x,y$ are on the same side of the median $m$ (including one of them being the median),
the one closer to $m$ has smaller social cost.
\end{lemma}

\begin{emptyproof}
Without loss of generality, assume that $x \in L \cup \SET{m}$, 
$y \in L$, and $d_{x,m} < d_{y,m}$.
Intuitively, $x$ has smaller social cost because more than half of the
population need to first get to $x$ before they can get to $y$.
Formally, we have
\begin{align*}
\Cost{x} & = \sum_{i \in L} \CProb{i} \Dist{i}{x} 
           + \sum_{i \in \{m\} \cup R} \CProb{i} \Dist{i}{x}
\; = \;  \sum_{i \in L} \CProb{i} \Dist{i}{x} 
       + \sum_{i \in \SET{m} \cup R} \CProb{i} \left(\Dist{i}{y} - \Dist{x}{y}\right) \\
& \stackrel{\CProb{L} \leq \CProb{m} + \CProb{R}}{\le}
  \sum_{i \in L} \CProb{i} \left(\Dist{i}{x} - \Dist{x}{y}\right) 
+ \sum_{i \in \SET{m} \cup R} \CProb{i} \Dist{i}{y} \\
& \stackrel{\bigtriangleup-\text{inequality}}{\leq}
  \sum_{i \in L} \CProb{i} \Dist{i}{y}
+ \sum_{i \in \SET{m} \cup R} \CProb{i} \Dist{i}{y}
\; = \; \Cost{y}. \QED
\end{align*}
\end{emptyproof}

As a simple corollary of Lemmas~\ref{lem:voteOrder} and \ref{lem:socialOrder},
  notice that if two candidates $(x,y)$ are drawn from the same side of the median
  (including when one of them is the median),
  majority voting always elects the socially better candidate.
This observation allows us to simplify the expression for \Social{\DIST}{\CPROB}{\CPROB} on the line,
\begin{align*}
\Social{\LINE}{\CPROB}{\CPROB} 
& = \sum_{i \in [n]} \CProb{i}^2 + \sum_{i,j\in [n]} 2\CProb{i}\CProb{j}\Ratio{i}{j}
  \; = \; 1 + \sum_{i\in L, j\in R} 2 \CProb{i}\CProb{j} (\Ratio{i}{j} - 1).
\end{align*}

\subsection{Proof of the Upper Bound of $4-2\sqrt{2}$}
\label{sec:line-ub}

In this section, we prove Theorem~\ref{thm:lineWorst}, 
showing that the worst-case distortion on the line is $4-2\sqrt{2}$.
The high-level idea is that, given any instance $(\LINE, \CPROB)$ with
support size larger than 3,
we can iteratively reduce its support size to 3 using a series of operations
(Lemmas~\ref{lem:mergeY}, \ref{lem:reduceR} and \ref{lem:reduceL}),
while preserving (or increasing) $\Social{\LINE}{\CPROB}{\CPROB}$.
Once the instance has support size 3,
we can optimize the locations and probabilities of these 3 points.

As before, let $m$ be the index of the median, 
and let $L = \SET{1,\ldots,m-1}$ and $R = \SET{m+1,\ldots,n}$
denote the points to the left and to the right of the median, respectively.
We can assume that both $L$ and $R$ are non-empty;
otherwise, the median is the leftmost or rightmost point,
and we always elect the socially better candidate.

The proof proceeds by moving probability mass within $L$ or within
$R$ to merge points until $|L| = |R| = 1$. 
None of the operations in this section will change the median $m$,
so the election results are still decided by the candidates' distance to $m$.

\begin{figure}[h]
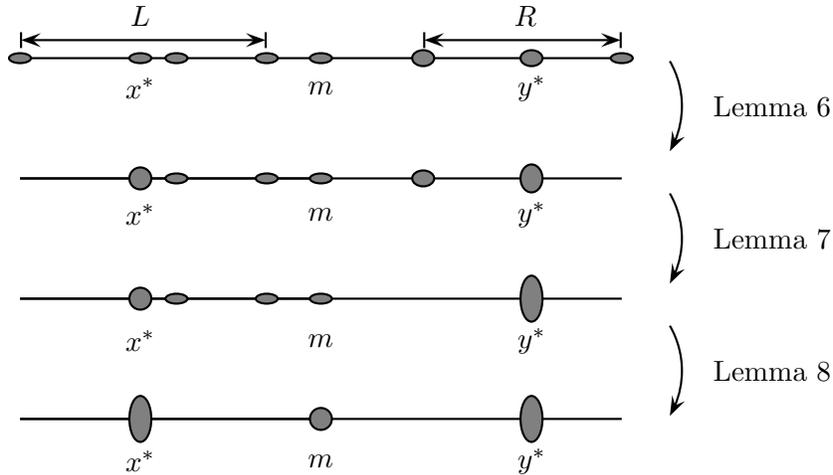

\centering
\psset{unit=0.8cm,arrowsize=0.2 1,fillstyle=solid}
\pspicture(-6,4)(9,-4)

\psline{-}(-5,3)(0,3)
\psline{-}(0,3)(5,3)
\psellipse[fillcolor=gray](-5,3)(0.2,0.1)
\psellipse[fillcolor=gray](-3,3)(0.2,0.1)
\psellipse[fillcolor=gray](-2.4,3)(0.2,0.1)
\psellipse[fillcolor=gray](-0.9,3)(0.2,0.1)
\psellipse[fillcolor=gray](0,3)(0.2,0.1)
\psellipse[fillcolor=gray](1.7,3)(0.2,0.15)
\psellipse[fillcolor=gray](3.5,3)(0.2,0.15)
\psellipse[fillcolor=gray](5,3)(0.2,0.1)

\psline{|<->|}(-5,3.3)(-0.9,3.3)
\psline{|<->|}(1.7,3.3)(5,3.3)
\rput(-3,3.7){$L$}
\rput(3.4,3.7){$R$}
\rput(0,2.5){$m$}
\rput(-3,2.5){$x^*$}
\rput(3.5,2.5){$y^*$}

\psline{-}(-5,1)(5,1)
\psline{-}(-5,1)(0,1)
\psellipse[fillcolor=gray](-3,1)(0.2,0.2)
\psellipse[fillcolor=gray](-2.4,1)(0.2,0.1)
\psellipse[fillcolor=gray](-0.9,1)(0.2,0.1)
\psellipse[fillcolor=gray](0,1)(0.2,0.1)
\psellipse[fillcolor=gray](1.7,1)(0.2,0.15)
\psellipse[fillcolor=gray](3.5,1)(0.2,0.25)
\rput(0,0.4){$m$}
\rput(-3,0.4){$x^*$}
\rput(3.5,0.4){$y^*$}

\psline{-}(-5,-1)(5,-1)
\psline{-}(-5,-1)(0,-1)
\psellipse[fillcolor=gray](-3,-1)(0.2,0.2)
\psellipse[fillcolor=gray](-2.4,-1)(0.2,0.1)
\psellipse[fillcolor=gray](-0.9,-1)(0.2,0.1)
\psellipse[fillcolor=gray](0,-1)(0.2,0.1)
\psellipse[fillcolor=gray](3.5,-1)(0.2,0.4)
\rput(0,-1.7){$m$}
\rput(-3,-1.7){$x^*$}
\rput(3.5,-1.7){$y^*$}

\psline{-}(-5,-3)(5,-3)
\psline{-}(-5,-3)(0,-3)
\psellipse[fillcolor=gray](-3,-3)(0.2,0.4)
\psellipse[fillcolor=gray](0,-3)(0.2,0.2)
\psellipse[fillcolor=gray](3.5,-3)(0.2,0.4)
\rput(0,-3.7){$m$}
\rput(-3,-3.7){$x^*$}
\rput(3.5,-3.7){$y^*$}

\psarc{<-}(4.5,2.2){1.5}{-30}{30}
\psarc{<-}(4.5,0){1.5}{-30}{30}
\psarc{<-}(4.5,-2.2){1.5}{-30}{30}
\rput(7.5,2.2){Lemma~\ref{lem:mergeY}}
\rput(7.5,0){Lemma~\ref{lem:reduceR}}
\rput(7.5,-2.2){Lemma~\ref{lem:reduceL}}
\endpspicture

\caption{An example of the series of operations
  (Lemmas~\ref{lem:mergeY}, \ref{lem:reduceR} and \ref{lem:reduceL})
  used to reduce the support size to 3 on the line,
  while preserving or increasing
  $\Social{\LINE}{\CPROB}{\CPROB}$.
Probability mass is roughly represented by sizes of ellipses.}
\label{fig:merge_line}
\end{figure}

When shifting the probability mass, we will not be able to guarantee
that no pairwise election sees a decrease in distortion.
Instead, we use a more global argument to show that the operation
increases the distortion \emph{on average}.
We define $\Ratio{i}{}$ to be the expected distortion conditioned on
one of the candidates being $i$,
and the other candidate being drawn according to \CPROB, that is,
\begin{equation*}
  \Ratio{i}{} = \sum_j \CProb{j} \Ratio{i}{j}.
\end{equation*}
We will show that so long as \CProb{L}, \CProb{m}, and \CProb{R} remain the same,
\Social{\LINE}{\CPROB}{\CPROB} is a linear function of the average
distortion on one side of the median.
By Lemmas~\ref{lem:voteOrder} and \ref{lem:socialOrder}, 
the pairwise distortion can be larger than 1 only if two candidates
are on different sides of $m$; therefore, 
\begin{align*}
\Social{\LINE}{\CPROB}{\CPROB}
  & = 1 + 2 \sum_{i\in L, j\in R} \CProb{i}\CProb{j} (\Ratio{i}{j} - 1)
  \; = \; 1 + 2 \sum_{i\in L, j\in [n]} \CProb{i}\CProb{j} (\Ratio{i}{j} - 1)
  \; = \; 1 - 2 \CProb{L} + 2 \sum_{i\in L} \CProb{i} \Ratio{i}{},\\
\Social{\LINE}{\CPROB}{\CPROB}
  & = 1 - 2 \CProb{R} + 2 \sum_{i\in R} \CProb{i} \Ratio{i}{}.
\end{align*}

The two preceding equations formalize that whenever \CProb{L} and
\CProb{R} stay constant and $\sum_{i \in R} \CProb{i}\Ratio{i}{}$ 
(or $\sum_{i \in R} \CProb{i}\Ratio{i}{}$) does not decrease,
\Social{\LINE}{\CPROB}{\CPROB} also does not decrease. 
This fact is exploited repeatedly in the proofs of the following
lemmas.

\begin{lemma} \label{lem:mergeY}
Let $y^* = \argmax_{y \in R} \Ratio{y}{}$ be the ``worst'' candidate in $R$.
Then, moving all probability mass from indices $y > y^*$ to $y^*$
does not decrease \Social{\LINE}{\CPROB}{\CPROB}.
(A symmetric claim holds for the worst candidate 
$x^* = \argmax_{x \in L} \Ratio{x}{}$.)
\end{lemma}

\begin{emptyproof}
Since the operation does not change \CProb{L} or \CProb{R},
it is sufficient to show that $\sum_{y \in R} p_y r_y$ does not decrease.
By Lemma~\ref{lem:voteOrder}, all election results 
between pairs $i, j \leq y^*$ are preserved.
Let $\CProbP{i}, \CostP{i},$ and $\RatioP{i}{}$ denote the
corresponding values of $\CProb{i}, \Cost{i}$ and $\Ratio{i}{}$ after
the operation. Then, for all $1 \le i,j \le y^*$,
\[
\RatioP{i}{j} 
\; = \; \frac{\CostP{\Winner{i}{j}}}{\CostP{\Opt{i}{j}}}
\; = \; \frac{\Cost{\Winner{i}{j}} - \sum_{y>y^*}\CProb{y}\Dist{y}{y^*}}{\Cost{\Opt{i}{j}} - \sum_{y>y^*}\CProb{y}\Dist{y}{y^*}}
\; \ge \; \frac{\Cost{\Winner{i}{j}}}{\Cost{\Opt{i}{j}}} 
\; = \; \Ratio{i}{j}.
\]

After the shift of probability mass, $y^*$ is the largest index.
Consider $m < y \leq y^*$.
Using that elections between two candidates on the same side of the
median always result in the socially better candidate winning, we bound
\[
\RatioP{y}{} 
\; = \; \sum_{1 \le i \le y^*} \CProbP{i} \RatioP{i}{y}
\; = \; \sum_{i \in L} \CProb{i} \RatioP{i}{y} + (1 - \CProb{L}) \cdot 1
\; \ge \; \sum_{i \in L} \CProb{i} \Ratio{i}{y} + (1 - \CProb{L}) 
\; = \; \Ratio{y}{}.
\]

Any candidates that used to be at $y > y^*$ are now at $y^*$, 
and $y^*$ used to be the worst candidate in $R$.
Hence, for all of the probability mass from locations $y > y^*$,
the expected distortion also weakly increases.
Combining these two cases, we get
\[ \sum_{m < y \le y^*} \CProbP{y} \RatioP{y}{}
\; = \; \sum_{m < y < y^*} \CProb{y} \RatioP{y}{} + \sum_{y^* \le y \le n} \CProb{y} \RatioP{y^*}{}
\; \ge \; \sum_{m < y \le y^*} \CProb{y} \Ratio{y}{} + \sum_{y^* \le y \le n} \CProb{y} \Ratio{y^*}{}
\; \ge \; \sum_{y \in R} \CProb{y} \Ratio{y}{}. \QED
\]
\end{emptyproof}

Lemma~\ref{lem:mergeY} can be applied repeatedly unless
the two worst candidates $x^*$ and $y^*$ are the leftmost and rightmost points.
We next show that in that case, 
either all the probability mass of $L$ or all the probability mass of
$R$ can be moved to $x^*$ or $y^*$, respectively.

\begin{lemma} \label{lem:reduceR}
Let $x^*$ and $y^*$ be the worst candidates in $L$ and $R$, respectively.
Assume w.l.o.g.~that $\Dist{m}{x^*} < \Dist{m}{y^*}$.
If $x^* = 1$ and $y^* = n$, 
then moving all probability mass from $R$ to $y^*$ 
does not decrease $\Social{\LINE}{\CPROB}{\CPROB}$.
\end{lemma}

\begin{proof}
As for the previous lemma, because we are only shifting probability
mass within $R$, it is sufficient to show that 
$\sum_{y \in R} \CProb{y} \Ratio{y}{}$ does not decrease.
Because more probability mass moved closer to $y^*$,
we have that $\CostP{y^*} \leq \Cost{y^*}$,
and because probability mass moved away from $L$ (to the right),
we get that $\CostP{i} \geq \Cost{i}$ for all $i \in L \cup \SET{m}$.

%
By Lemma~\ref{lem:voteOrder}, $y^*$ loses all of his elections both
before and after the move. 
Moreover, by Lemma~\ref{lem:socialOrder}, we get  
$\Ratio{y^*}{} = (1- \CProb{L}) + \sum_{i \in L} \CProb{i} \Ratio{i}{y^*}$
before the move, and  
$\RatioP{y^*}{} = (1- \CProbP{L}) + \sum_{i \in L} \CProbP{i} \RatioP{i}{y^*} 
= (1- \CProb{L}) + \sum_{i \in L} \CProb{i} \RatioP{i}{y^*}$ after the move.  
Since $\RatioP{i}{y^*} = \frac{\CostP{i}}{\CostP{y^*}} \geq \frac{\Cost{i}}{\Cost{y^*}} 
= \Ratio{i}{y^*}$ for all $i \in L$, 
we get that $\RatioP{y^*}{} \geq \Ratio{y^*}{}$.
Finally, because $y^*$ used to be the worst candidate in $R$,
and after the move of probability mass is the only candidate in $R$,
we bound 
\[
\CProbP{y^*} \RatioP{y^*}{} 
\; = \; \sum_{y\in R} \CProb{y} \RatioP{y^*}{} 
\; \ge \; \sum_{y\in R} \CProb{y} \Ratio{y^*}{} 
\; \ge \; \sum_{y\in R} \CProb{y} \Ratio{y}{},
\]
which concludes the proof.
\end{proof}

Once neither Lemma~\ref{lem:mergeY} nor Lemma~\ref{lem:reduceR}
can be applied, we can apply Lemma~\ref{lem:reduceL}.

\begin{lemma} \label{lem:reduceL}
Let $x^* = 1$, $y^* = n$ be the worst candidates in $L$ and $R$, respectively.
If $|L| > 1$, $|R| = 1$ and $\Dist{m}{y^*} > \Dist{m}{x^*}$,
then the size of $L$ can be reduced by 1 
without decreasing \Social{\LINE}{\CPROB}{\CPROB}.
\end{lemma}

\begin{emptyproof}
Notice that $m = n-1$ and $L = \SET{1, \ldots, n-2}$.
Recall that the only elections in which the winner could be socially
inferior are those involving $n$ and a candidate $x \in L$.
Also, because $\Dist{m}{n} > \Dist{m}{x^*} \geq \Dist{m}{i}$ for all
$i$, we obtain that $n$ loses all elections.
We split the proof into two cases.

\begin{enumerate}
\item If there exists an $i \in L$ with $\Cost{i} \leq \Cost{n}$, then
  in particular, $\Cost{n-2} \leq \Cost{n}$.
  Thus, candidate $n-2$ wins all elections against $i \leq n-2$ (as he
  should) and against $n$ (as he should), while losing to $m$ (as he
  should). This implies that $\Ratio{n-2}{} = 1$.

  Consider the effect of moving all probability mass from $n-2$ to the
  median $m=n-1$.
  First, all election results remain the same.
  The contribution of the probability mass that used to be at $n-2$ 
  to the distortion does not change. (It was 1 before and is still 1.)
  Furthermore, \Cost{n} decreases while \Cost{i} increases for all
  $i < n-2$. 
  Because $n$ loses all pairwise elections, the overall distortion
  can only increase.

\item If $\Cost{i} > \Cost{n}$ for all $i \in L$,
  the expected distortion is exactly
\begin{align*}
\Social{\LINE}{\CPROB}{\CPROB} 
& = 1 - 2 \CProb{n} \CProb{L} + 2 \CProb{n} \sum_{i\in L} \CProb{i} \frac{\Cost{i}}{\Cost{n}}.
\end{align*}
Let $x_i$ denote the position of point $i$ on the line.
Since $|L| > 1$, we have a point at position $x_2$ in $L$ with $x_1 < x_2 < x_3$.
Writing $Y := \sum_{j \neq 2} \CProb{j} |x_n - x_j|$ and
$X_i := \sum_{j \neq 2} \CProb{j} |x_i - x_j|$, 
we get that $\Cost{i} = X_i + \CProb{2} |x_i - x_2|$,
and $\Cost{n} = Y + \CProb{2} (x_n - x_2)$.
Hence, we can rewrite 
\begin{align*}
\Social{\LINE}{\CPROB}{\CPROB} 
& = 1 - 2 \CProb{n} \CProb{L} + 2 \CProb{n} \sum_{i\in L} \CProb{i}
  \frac{X_i + \CProb{2} |x_i - x_2|}{Y + \CProb{2} (x_n - x_2)}\\
& = 1 - 2 \CProb{n} \CProb{L}
  + \frac{2 \CProb{n}}{Y + \CProb{2} x_n - \CProb{2} x_2}
  \cdot \left(\sum_{i\in L} \CProb{i} X_i
  - \CProb{2} (\CProb{1} x_1 - \sum_{i=3}^{n-2} \CProb{i} x_i)
  + \CProb{2} x_2 (\CProb{1} - \sum_{i=3}^{n-2} \CProb{i})\right)\\
& = 1 - 2 \CProb{n} \CProb{L}
  + \frac{2 \CProb{n}}{(Y/\CProb{2} + x_n) - x_2}
  \cdot \left(\sum_{i\in L} \CProb{i} X_i/\CProb{2}
  - (\CProb{1} x_1 - \sum_{i=3}^{n-2} \CProb{i} x_i)
  + x_2 (\CProb{1} - \sum_{i=3}^{n-2} \CProb{i})\right).
\end{align*}
Treating everything except $x_2$ as constant,
  this expression is of the form $\frac{B + \beta x_2}{A-x_2}$ for all $x_2 \in [x_1, x_3]$,
  where $A$, $B$, and $\beta$ are constants independent of $x_2$.
The derivative of this expression with respect to $x_2$ is
$\frac{\beta A + B}{(A-x_2)^2}$;
its sign is always the sign of $\beta A + B$.
If $\beta A + B > 0$, then increasing $x_2$ to $x_3$ strictly increases the expected distortion;
  otherwise, $x_2$ can be decreased to $x_1$ without decreasing the expected distortion.
In either case, we reduce the size of $L$ by 1. \QED
\end{enumerate}
\end{emptyproof}

We are now ready to prove Theorem~\ref{thm:lineWorst}.

\begin{extraproof}{Theorem~\ref{thm:lineWorst}}
By Lemmas~\ref{lem:mergeY}, \ref{lem:reduceR} and \ref{lem:reduceL},
the worst-case instance $(\LINE, \CPROB, \CPROB)$ has support size (at
most) 3.
Let $x_1 \leq x_2 \leq x_3$ be the locations on the line.
By rescaling and mirroring, we may assume 
without loss of generality that $x_1 = 0$, $x_3 = 1$, and $x_2 > \half$.

If $x_2$ were not the median of the distribution,
then the socially better candidate would always win,
giving $\Social{\LINE}{\CPROB}{\CPROB} = 1$.
So in a worst-case distribution, $x_2$ must be the median,
and the socially worse candidate must win the election between $x_1$ and $x_3$.
Because $x_2 > \half$, $x_3$ is closer to the median, so he wins the
election between $x_1$ and $x_3$; therefore, $x_1$ must have lower
cost than $x_3$.
The expected distortion is
\[
\Social{\LINE}{\CPROB}{\CPROB} 
\; = \; (1 - 2 \CProb{1}\CProb{3}) \cdot 1
   + 2 \CProb{1}\CProb{3} \cdot \frac{\Cost{3}}{\Cost{1}} 
\; = \; (1 - 2 \CProb{1}\CProb{3}) 
   + 2 \CProb{1}\CProb{3} \cdot \frac{\CProb{1} + \CProb{2} (1-x_2)}{\CProb{2} x_2 + \CProb{3}}.
\]
This expression is monotonically decreasing in $x_2$
and monotonically increasing in $\CProb{1}$, 
so it is maximized when we take the limit 
$x_2 \to \half$ and $\CProb{1} \to \half$.
In particular,
\[ 
\Social{\LINE}{\CPROB}{\CPROB} 
\; \leq \; (1 - \CProb{3}) + \CProb{3} \cdot \frac{1/2 + \CProb{2}/2}{\CProb{2}/2 + \CProb{3}}
\; = \; (1 - \CProb{3}) + \CProb{3} \cdot \frac{3-2\CProb{3}}{1+2\CProb{3}}, \]
which is maximized at $\CProb{3} = \frac{\sqrt{2}-1}{2}$ 
(as in Example~\ref{ex:line-iid}),
where it attains a value of $4-2\sqrt{2}$.
\end{extraproof}


\section{Different Distributions}
\label{sec:different}

In this section, we prove a tight bound of 2 on the worst-case
distortion of voting, when two candidates are drawn
i.i.d.~from a distribution \CPROB which may be different from
the voter distribution~\VPROB.
This ratio is tight for both general metric spaces and the line,
and the lemmas we prove in this section apply to arbitrary metric spaces.

We begin with an example on the line (a variant of
Example~\ref{ex:line-three}) which establishes the lower bound of 2.
The candidate distribution \CPROB has probability 1/2 at position $-1$,
and the other $1/2$ at position 1.
The voter distribution \VPROB has a $(1/2 - \eps)$ fraction of the
voters at position $-1$, while the remaining voters are just to the
right of center at position $\eps > 0$.
With probability 1/2, we draw two different candidates, 
and the distortion is $3-O(\eps)$;
otherwise, we draw two candidates from the same location, 
getting a distortion of 1.
Therefore, the expected distortion of the instance is 
$2-O(\eps) \to 2$ as $\eps \to 0$.

The challenge is to establish the matching upper bound.
In proving the upper bound, some of the techniques we establish
  will be useful in Section~\ref{sec:metric}.

\begin{theorem} \label{thm:diff-dist}
For all instances $(\DIST, \CPROB, \VPROB)$, 
the expected distortion $\Social{\DIST}{\CPROB}{\VPROB}$ is at most $2$.
\end{theorem}

The overall proof structure is as follows.
First, we show in Lemma~\ref{lem:cost-bound} that if 
$i = \Winner{i}{i'}$, then $\Cost{i} \leq 3\Cost{i'}$. 
That is, while the election winner can be socially worse, he cannot be
too much worse.\footnote{Lemma~\ref{lem:cost-bound} is a
special case of the more general result \cite[Theorem 4]{anshelevich:bhardwaj:postl};
we present a self-contained proof here for completeness.}
Lemma~\ref{lem:cost-bound} is the only place where we use the metric
structure and the voter distribution.
Subsequently, we rewrite the social cost function
\Social{\DIST}{\CPROB}{\VPROB} accordingly, 
and then treat the costs as completely arbitrary numbers.

Second, in Lemma~\ref{lem:diff-with-cap},
we prove that if all pairwise elections have
distortion at most $1 \le \alpha \le 3$,
then $\Social{\DIST}{\CPROB}{\VPROB} \le (1+\alpha)/2$.
(While in this section, we will only use the lemma with
  $\alpha = 3$, the version with general $\alpha$ constitutes a key
  step in Section~\ref{sec:metric}.)

\begin{lemma}[\cite{anshelevich:bhardwaj:postl}]
\label{lem:cost-bound}
Let $i = \Winner{i}{i'}$. Then, $\Cost{i} \leq 3\Cost{i'}$.
\end{lemma}

\begin{emptyproof}
In the following derivation, we will use that:
\begin{itemize}
\item Because $i$ beats $i'$, at least half of the voters are at least
  as close to $i$ as to $i'$.
\item For any voter $j$ who is at least as close to $i$ as to $i'$, 
  the triangle inequality implies that 
$\Dist{i'}{i} \leq \Dist{i'}{j} + \Dist{j}{i} \leq 2 \Dist{i'}{j}$.
\end{itemize}
Then, we can bound \Cost{i} as follows:
\begin{align*}
\Cost{i}
& = \sum_{j: \Dist{i}{j} \leq \Dist{i'}{j}} \VProb{j} \cdot \Dist{i}{j}
  + \sum_{j: \Dist{i}{j} > \Dist{i'}{j}} \VProb{j} \cdot \Dist{i}{j} \\
& \stackrel{\bigtriangleup-\text{inequality}}{\leq}
    \sum_{j: \Dist{i}{j} \leq \Dist{i'}{j}} \VProb{j} \cdot \Dist{i'}{j}
  + \sum_{j: \Dist{i}{j} > \Dist{i'}{j}} \VProb{j} \cdot (\Dist{i'}{j} + \Dist{i}{i'})\\
& \stackrel{i \text{ beats } i'}{\leq}
    \sum_{j: \Dist{i}{j} \leq \Dist{i'}{j}} \VProb{j} \cdot (\Dist{i'}{j} + \Dist{i}{i'})
  + \sum_{j: \Dist{i}{j} > \Dist{i'}{j}} \VProb{j} \cdot \Dist{i'}{j}\\
& \leq 
    \sum_{j: \Dist{i}{j} \leq \Dist{i'}{j}} \VProb{j} \cdot (3 \Dist{i'}{j})
  + \sum_{j: \Dist{i}{j} > \Dist{i'}{j}} \VProb{j} \cdot \Dist{i'}{j}\\
& \leq 3 \Cost{i'}.\QED
\end{align*}
\end{emptyproof}

\begin{lemma}
\label{lem:diff-with-cap}
For any $1 \le \alpha \le 3$ and any instance $(\DIST, \CPROB, \VPROB)$,
  if  $\Ratio{i}{j} = \frac{\Cost{\Winner{i}{j}}}{\Cost{\Opt{i}{j}}} \le \alpha$ for all $(i, j)$,
  then $\Social{\DIST}{\CPROB}{\VPROB} \le \frac{1+\alpha}{2}$.
\end{lemma}
\begin{emptyproof}

Consider an instance $(\DIST, \CPROB, \VPROB)$ and its associated costs \COST. 
Without loss of generality, assume that 
$\Cost{1} \leq \Cost{2} \leq \cdots \leq \Cost{\NUMCAND}$.
For each candidate $i$, let $\Last{i} = \max \Set{j}{\Cost{j} \leq \alpha \Cost{i}}$.
Notice that by the assumption that $\Ratio{i}{j} \le \alpha$ for all $i,j$,
whenever $j > \Last{i}$, 
we have that $\Winner{i}{j} = \Opt{i}{j}$, resulting in a cost ratio of 1. 
We can therefore bound the expected distortion (minus 1) as follows:
\begin{align}
\Social{\DIST}{\CPROB}{\VPROB} - 1
& \leq 2 \sum_{i < j \leq \Last{i}} \CProb{i} \CProb{j} \cdot \left(\frac{\Cost{j}}{\Cost{i}} - 1\right)
\; =: \; \CSoc{\CPROB}{\COST}.
\label{eqn:csoc}
\end{align}

The upper bound \CSoc{\CPROB}{\COST} assumes that the worse candidate
wins whenever the two candidates' social costs are within a factor of
$\alpha$ of each other.
Note that this upper bound \CSoc{\CPROB}{\COST} makes no more
reference to distances or voter distributions.
It depends on a distribution over candidates and a cost vector,
both of which can be arbitrary,
and it assumes that all elections whose candidates' costs are more
than a factor $\alpha$ apart choose the socially better candidate,
while all other elections choose the socially worse candidate.

We will now argue that \CSoc{\CPROB}{\COST} is at most $\frac{\alpha - 1}{2}$. 
First, we show that the expression is maximized by moving
probability mass so that $c_i$ and $c_j$ are at most a factor $\alpha$
apart for every $i$ and $j$ in the support of \CPROB. 
Suppose that there exists a pair $i < j$ in the support of \CPROB with 
$j > \Last{i}$, i.e., with $\Cost{j} > \alpha \Cost{i}$.
Consider moving $\epsilon$ probability mass from \CProb{i} to \CProb{j},
where a negative value of $\epsilon$ moves probability mass from
\CProb{j} to \CProb{i}; 
call the resulting probability vector $\CPROB(\epsilon)$.
Because our choice of $i$ and $j$ avoids the bilinear term $p_i p_j$ in
\eqref{eqn:csoc}, 
\CSoc{\CPROB(\epsilon)}{\COST} is a linear function of $\epsilon$.
Therefore, the expression is maximized at an extreme, 
i.e., by moving all the probability mass from one of $i$
and $j$ to the other. 

Once all points in the support of \CPROB are at most a
factor $\alpha$ apart in social cost, 
the expression for \CSoc{\CPROB}{\COST} in \eqref{eqn:csoc} becomes a
sum over all pairs of points. 
Assume that the support of \CPROB has size $n' \geq 3$, and associated
costs $\Cost{1} < \Cost{2} < \cdots < \Cost{n'}$. 
(The inequalities can be assumed to be strict, because two points
$i,i'$ with the same cost can be merged without affecting the value
\CSoc{\CPROB}{\COST}.)
Considering all terms except \Cost{2} as constants, 
\CSoc{\CPROB}{\COST} is of the form
$\beta_1 + \beta_2 \Cost{2} + \beta_3/\Cost{2}$
(with $\beta_2, \beta_3 \geq 0$),
which is convex in \Cost{2}.
In particular, it attains its maximum at
$\Cost{2} = \Cost{1}$ or $\Cost{2} = \Cost{3}$.
In either case, we can merge the probability mass of point 2 with 
1 or 3, reducing the support size by 1 without decreasing
\CSoc{\CPROB}{\COST}.
By repeating such merges, we eventually arrive at a distribution with
support size $2$ and $\Cost{2} \leq \alpha \Cost{1}$.
Finally, we can bound
\[
\Social{\DIST}{\CPROB}{\VPROB}
\; = \; 1 + \CSoc{\CPROB}{\COST}
\; \leq \; 1 + 2 \CProb{1} (1-\CProb{1}) \cdot (\alpha - 1)
\; \leq \; 1 + \half (\alpha-1)
\; = \; \frac{1+\alpha}{2}.\QED
\]
\end{emptyproof}




\newcommand{\RatioBoverI}{\frac{1-\CProb{}+\Radius{}+\CProb{}\Radius{}}{\CProb{}(1-\Radius{})}}
\newcommand{\RatioIoverB}{\frac{3+\Radius{}}{1-\Radius{}}}
\newcommand{\RatioIwinsOverB}{\frac{2}{1-\Radius{}}}
\newcommand{\RatioIwinsOverJ}{\RatioIwinsOverB \cdot \RatioBoverI}

\section{Identical Distributions in General Metric Spaces}
\label{sec:metric}

In this section, we examine the setting where the underlying metric
space is arbitrary, and the candidates are drawn independently from
the population of voters.
We establish the following main theorem:

\begin{theorem} \label{thm:main-general-metric}
The worst-case distortion 
$\sup_{(\DIST, \CPROB, \CPROB)} \Social{\DIST}{\CPROB}{\CPROB}$
is between $\frac{3}{2}$ and $2 - \frac{1}{652}$.
\end{theorem}

Key to the upper bound portion of this theorem is the following lemma.
\begin{lemma} \label{lem:two-minus-eps}
Assume that $\delta \leq \frac{1}{100}$.
Let $(\DIST, \CPROB, \CPROB)$ be an instance with maximum pairwise
distortion (exactly) $3-\delta$.
Then, $\Social{\DIST}{\CPROB}{\CPROB} \le \frac{3}{2} + 9 \sqrt{\delta}$.
\end{lemma}
We prove Lemma~\ref{lem:two-minus-eps} in Section~\ref{pf:two-minus-eps}.
That proof relies on the following structural characterization:
if a pair of candidates has distortion $3-\delta$ for sufficiently
small $\delta$, then the instance must be very structured:
nearly half the probability mass must be
concentrated very close to the socially optimal candidate,
and most of the remaining candidates must be nearly equidistant to
the two candidates.

\begin{extraproof}{Theorem~\ref{thm:main-general-metric}}
We begin by proving the lower bound, by constructing a family of
instances whose distortion converges to $\frac{3}{2}$.
We label the $\NUMCAND+1$ points $\{0, 1, \ldots, \NUMCAND\}$.
We set $\CProb{0} = \frac{1-\eps}{2}$, 
and all other $\CProb{i} = \frac{1+\eps}{2\NUMCAND}$.
The distances\footnote{To avoid tie breaking issues, consider the
  distances as perturbed by distinct and very small amounts.} 
are $\Dist{0}{i} = 1$ for all $i > 0$, 
and $\Dist{i}{j} = 1-\eps$ for all $i,j > 0$.
See Figure \ref{fig:metric_worst} for an illustration.

\begin{figure}[h]
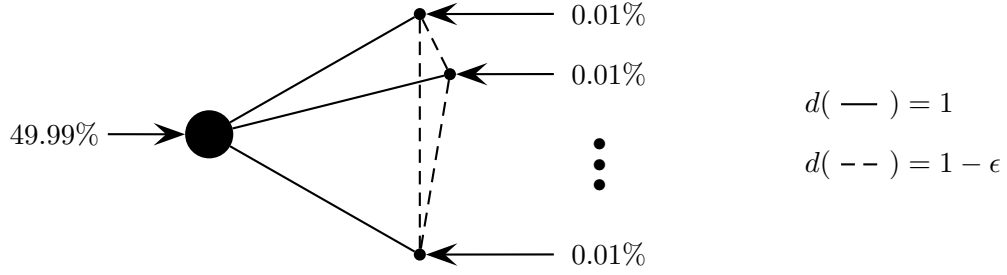

\centering
\psset{unit=0.8cm,arrowsize=0.2 5}
\pspicture(-7,-2.5)(7,3)
\cnode[fillstyle=solid,fillcolor=black](-4.5,0){0.4}{Big}
\cnode[fillstyle=solid,fillcolor=black](-1,2){0.1}{Sml1}
\cnode[fillstyle=solid,fillcolor=black](-0.5,1){0.1}{Sml2}
\cnode[fillstyle=solid,fillcolor=black](-1,-2){0.1}{Sml3}
\ncline{Big}{Sml1}
\ncline{Big}{Sml2}
\ncline{Big}{Sml3}
\ncline[linestyle=dashed]{Sml1}{Sml2}
\ncline[linestyle=dashed]{Sml3}{Sml2}
\ncline[linestyle=dashed]{Sml1}{Sml3}

\rput(-7,0){\rnode{BigP}{49.99\%\;}}
\ncline[arrows=->]{BigP}{Big}
\rput(2,2){\rnode{Sml1P}{\; 0.01\%}}
\ncline[arrows=->]{Sml1P}{Sml1}
\rput(2,1){\rnode{Sml2P}{\; 0.01\%}}
\ncline[arrows=->]{Sml2P}{Sml2}
\rput(2,-2){\rnode{Sml3P}{\; 0.01\%}}
\ncline[arrows=->]{Sml3P}{Sml3}
\rput{90}(2,-0.5){$\bullet\bullet\bullet$}

\rput[l](5.4,0.5){$d(\qquad)=1$}
\rput[l](5.4,-0.5){$d(\qquad)=1-\eps$}
\psline(6,0.5)(6.6,0.5)
\psline[linestyle=dashed](6,-0.5)(6.6,-0.5)

\endpspicture

\caption{A class of instances for general metric spaces in which \Social{\DIST}{\CPROB}{\CPROB} approaches $\frac{3}{2}$.}
\label{fig:metric_worst}
\end{figure}

This way, all voters/candidates in the set $\{1,\ldots,\NUMCAND\}$
prefer each other over the voter/candidate $0$. 
Therefore, even though candidate $0$ is socially optimal 
(with a cost $\Cost{0} = \half + O(\epsilon))$,
he loses to any other candidate in the election;
the other candidates' costs are 
$\Cost{i} = 1 - O(1/\NUMCAND) - O(\epsilon)$.

With probability $\half - O(\epsilon)$, an election occurs between
candidate 0 and some other candidate $i > 0$, resulting in distortion
$2 - O(\epsilon) - O(1/\NUMCAND)$.
In the other cases (two candidates from 0, or two candidates $i, j >
0$), the distortion is at least 1.
Hence, the overall expected distortion is at least
$(\half - O(\epsilon)) \cdot (2 - O(\epsilon) - O(1/\NUMCAND)) 
+ \half \cdot 1 = \frac{3}{2} - O(\epsilon) - O(1/\NUMCAND)$.
As $\eps \to 0$ and $\NUMCAND \to \infty$, the distortion approaches
$\frac{3}{2}$.

For the upper bound, let $\delta = \frac{1}{326}$ and consider the following two cases.
If all pairwise elections have distortion at most $3-\delta$,
then Lemma~\ref{lem:diff-with-cap} implies that 
the overall expected distortion $\Social{\DIST}{\CPROB}{\CPROB}$ is at
most $2-\delta/2 = 2-\frac{1}{652}$. If some pair of candidates has distortion at least $3-\delta$, then Lemma~\ref{lem:two-minus-eps} implies that the overall expected distortion is at most $\frac{3}{2} + 9 \sqrt{\delta} \leq 2-\frac{1}{652}$. Together, these two cases complete the proof of the theorem.
\end{extraproof}

As mentioned above, the key insight in the proof of
Lemma~\ref{lem:two-minus-eps} is that when a pair of candidates has
$\Ratio{x}{y} \geq 3 - \delta$, \emph{nearly} half the probability
mass must be concentrated \emph{very close to} the socially optimal
candidate, and \emph{most} of the remaining candidates must be
\emph{nearly} equidistant to the two candidates.
Trading off these four sources of approximation makes the proof of
the lemma fairly complex.
To illustrate the key ideas more cleanly, we therefore begin by
proving the following special case of Lemma~\ref{lem:two-minus-eps}
with $\delta=0$.

\begin{lemma} \label{lem:given-half-point}
Let $(\DIST, \CPROB, \CPROB)$ be an instance.
If there exists a pair of candidates $x,y$ with 
$\frac{\Cost{\Winner{x}{y}}}{\Cost{\Opt{x}{y}}} = 3$, 
then $\Social{\DIST}{\CPROB}{\CPROB} = 1.5$.
\end{lemma}

As before, we let $\CProb{A} = \sum_{i \in A} \CProb{i}$ denote the
total probability mass in $A$.
In addition, throughout this section,
$\CPROB_{A}$ is the conditional candidate/voter distribution given
that candidate $i$ is drawn from $A$;
that is, $(\CPROB_{A})_i = \CProb{i} / \CProb{A}$.
We use $\Dist{i}{A}$ to denote the \emph{average} distance from $i$ to
the set $A$, i.e.,
$\Dist{i}{A} = \Expect[j \sim \CPROB_A]{\Dist{i}{j}}$.

\subsection{Proof of Lemma~\ref{lem:given-half-point}}
Assume that $y = \Winner{x}{y}$ and $x = \Opt{x}{y}$.
We assume without loss of generality that $\Dist{x}{y} = 2$.
The fact that $\Cost{y} = 3 \Cost{x}$ implies very stringent
conditions on the instance: we will begin by showing that
half of the probability mass must be at $x$, $x$ is socially optimal,
and all other locations are at distance\footnote{In a sense, this
  extreme example \emph{does} rely on tie breaking. Since we are
  proving an \emph{upper} bound here, this is not a concern.}
1 from $x$ and $y$.

Let $Y$ be the set of voters preferring $y$ over $x$, and
$X = \Compl{Y}$ the set of voters preferring $x$ over $y$.
Then, we can bound
\begin{align*}
\Cost{y} 
& = \CProb{Y} \Dist{y}{Y} + \CProb{X} \Dist{y}{X} \\
& \stackrel{\bigtriangleup-\text{inequality}}{\leq} 
\CProb{Y} \Dist{y}{Y} + \CProb{X} (\Dist{y}{x} + \Dist{x}{X}) \\
& \stackrel{y \text{ beats } x}{\leq} 
\CProb{Y} (\Dist{y}{Y} + \Dist{y}{x}) + \CProb{X} \Dist{x}{X} \\
& \stackrel{\bigtriangleup-\text{inequality}}{\leq} 
\CProb{Y} (\Dist{y}{Y} + \Dist{y}{Y} + \Dist{x}{Y})
   + \CProb{X} \Dist{x}{X}\\
& \leq
3(\CProb{Y} \Dist{x}{Y} + \CProb{X} \Dist{x}{X})\\
& = 3\Cost{x}.
\end{align*}

Because $\Cost{y} = 3\Cost{x}$ by assumption, all of the inequalities
must be tight, which implies the following:
\begin{enumerate}
\item By the second (in)equality, $\CProb{Y} = \CProb{X} = \half$.
\item By the final (in)equality, $\Dist{x}{Y} = \Dist{y}{Y}$, so all
  points in $Y$ are equidistant from $x$ and $y$.
  Furthermore, because $\CProb{X} \Dist{x}{X} = 3 \CProb{X} \Dist{x}{X}$,
  we get $\Dist{x}{X} = 0$.
\item By the first (in)equality, $\Dist{y}{X} = \Dist{y}{x} + \Dist{x}{X}
  = 2$.
\item By the third (in)equality, $\Dist{y}{x} = \Dist{y}{Y} + \Dist{x}{Y}$,
so (because $\Dist{y}{Y} = \Dist{x}{Y}$ and by triangle
inequality), $\Dist{y}{i} = \Dist{x}{i} = 1$ for all $i \in Y$.
\end{enumerate}

Because $\Dist{x}{X} = 0$, we can write $\CProb{x} = \half$.
We then have that $\Cost{x} = \half$, and $\Cost{y} = \frac{3}{2}$.
Let $A$ denote the set of all candidates other than $x$.
The expected distortion is then
\begin{align*}
\Social{\DIST}{\CPROB}{\CPROB}
&= \CProb{x}^2 + 2 \CProb{x}\CProb{A} \expect{i \sim A}{\Ratio{i}{x}} + \CProb{A}^2 \expect{i,j \sim A}{\Ratio{i}{j}} \\
&= 1/4 + 1/4 \cdot \expect{i,j \sim A}{\Ratio{i}{x} + \Ratio{j}{x} + \Ratio{i}{j}}.
\end{align*}

Let $\Delta_{i,j} = \Ratio{i}{x} + \Ratio{j}{x} + \Ratio{i}{j}$.
We will show that $\Delta_{i,j} \le 5$ for all $i, j \in A$, 
and thus $\Social{\DIST}{\CPROB}{\CPROB} \le 1.5$.
The three key properties we exploit repeatedly are the following.

\begin{enumerate}
\item $x$ is socially optimal, i.e., $\Cost{i} \ge \Cost{x}$ for
  all $i \in A$.
  This is because $\Cost{x} = \half$, and
  for each $i \in A$, at least all voters at $x$ are at distance 1.
\item $\Cost{i} \le 3\Cost{x}$ for all $i \in A$. 
  This is because $\Dist{i}{A} \le \Dist{i}{x} + \Dist{x}{A} = 2$
  and $\Dist{i}{x} = 1$, giving a total cost of at most $\frac{3}{2}$.
\item If some $i \in A$ beats $x$, then $\Cost{i} \le 2\Cost{x}$. 
  This is because everyone in $A$ has to vote for $i$, implying that 
  $\Dist{i}{A} \le 1$, giving $\Cost{i} \leq 1$.
\end{enumerate}

Now fix some pair $i, j \in A$, and assume without loss of generality
that $i$ wins the election over $j$.
For each of the three elections that contribute to $\Delta_{i,j}$ ($i$
vs.~$x$, $j$ vs.~$x$, $i$ vs.~$j$), there is a term which is 
1 if the election chooses the socially better candidate
(e.g., $\Winner{i}{x} = \Opt{i}{x}$), and at most 3 otherwise.
Thus, if we ever had $\Delta_{i,j} > 5$, at least two of the three
elections would have to produce the socially worse candidate as a
winner, e.g., $\Winner{i}{x} \neq \Opt{i}{x}$.
We distinguish three possible cases.

\begin{enumerate}
\item If $x$ beats $j$, then $i$ must beat $x$ and (because we assumed
  $i$ to beat $j$) $j$ must have lower cost than $i$.
  Because $x$ is socially optimal (in particular having lower social
  cost than $j$), using that $\Cost{i} \leq 2\Cost{x}$, we have that
\[ \Delta_{i,j} = \Ratio{i}{x} + \Ratio{j}{x} + \Ratio{i}{j}
   = \frac{\Cost{i}}{\Cost{x}} + 1 + \frac{\Cost{i}}{\Cost{j}}
   \le 1 + \frac{\Cost{i}}{\Cost{x}} + \frac{\Cost{i}}{\Cost{x}} \le 1 + 2 + 2 = 5. \]

\item If $x$ beats $i$, then $j$ must beat $x$ and have lower cost
  than $i$. Then we obtain the expression.
\[ \Delta_{i,j} = \Ratio{i}{x} + \Ratio{j}{x} + \Ratio{i}{j} 
   = 1 + \frac{\Cost{j}}{\Cost{x}} + \frac{\Cost{i}}{\Cost{j}}. \]
Treating $\Cost{j}$ as a variable $t$, we have an expression of the
form $\frac{t}{\Cost{x}} + \frac{\Cost{i}}{t}$, which is convex and
hence maximized at an extreme point ($t=\Cost{i}$ or $t=\Cost{x}$), 
giving an upper bound of $1 + 1 + \frac{\Cost{i}}{\Cost{x}} \le 5$.

\item Finally, we have the case that both $i$ and $j$ beat $x$ 
(implying that $\Cost{i} \le 2\Cost{x}$ and $\Cost{j} \le 2\Cost{x}$).
If $i$ has lower social cost than $j$, we can bound
\[ \Delta_{i,j} = \Ratio{i}{x} + \Ratio{j}{x} + \Ratio{i}{j} \le 2 + 2 + 1 = 5. \]
Otherwise we have $\Cost{x} \le \Cost{j} < \Cost{i}$, and
obtain the expression
$\Delta_{i,j} = \frac{\Cost{i}}{\Cost{x}} + \frac{\Cost{j}}{\Cost{x}} + \frac{\Cost{i}}{\Cost{j}}$.
Again, the expression $\frac{\Cost{j}}{\Cost{x}} +
\frac{\Cost{i}}{\Cost{j}}$ is maximized at $\Cost{j} = \Cost{x}$ 
or $\Cost{j}=\Cost{i}$, in both cases giving us a bound of 
$\frac{\Cost{i}}{\Cost{x}} + 1 + \frac{\Cost{i}}{\Cost{x}} \le 5$. \QED
\end{enumerate}

\subsection{Proof of Lemma~\ref{lem:two-minus-eps}}
\label{pf:two-minus-eps}

The proof of Lemma~\ref{lem:two-minus-eps} follows the same roadmap 
as the proof of Lemma~\ref{lem:given-half-point},
except that we no longer have a point with probability mass 1/2.
Instead, close to half of the probability mass will be 
in a ball $B$ of small radius around $x$.
The three key properties used to bound $\Delta_{i,j}$ will then be
replaced with approximate (slightly inferior) versions.

For any pair of candidates $i, j$, we will be frequently using the
following upper bounds on \Cost{i}:
\begin{align}
\Cost{i} & \le \Cost{j} + \Dist{i}{j}, \label{eqn:triangle-costs}\\
\Cost{i} & \le \Cost{j} + \frac{\Dist{i}{j}}{2}
\quad \text{ whenever $i$ beats $j$.} \label{eqn:triangle-winner}
\end{align}
Inequality~\eqref{eqn:triangle-costs} is simply by the triangle inequality, 
while Inequality~\eqref{eqn:triangle-winner} also uses the fact that 
half of the voters are closer to $i$, and at most the remaining half
can contribute to the cost gap.

Let $(x,y)$ be the election maximizing $\Ratio{x}{y}$,
having $\Ratio{x}{y} = 3-\delta$.
Without loss of generality, assume that $y = \Winner{x}{y}$
and $\Dist{x}{y} = 2$.
Because $\Cost{y} = (3-\delta) \Cost{x}$ and 
$\Cost{y} \leq \Cost{x} + \frac{\Dist{x}{y}}{2} = \Cost{x} + 1$,
we obtain that 
\begin{align} \label{eqn:two-minus-eps-cost-x-ub}
\Cost{x} \le \frac{1}{2-\delta}.
\end{align}

As before, let $X$ be the set of voters closer to $x$ than to $y$, 
and $Y = \Compl{X}$ the set closer to $y$.
Then, $\CProb{X} \leq \half \leq \CProb{Y}$.
Following our intuition from the proof of
Lemma~\ref{lem:given-half-point}, 
we partition the points into three disjoint sets $A$, $B$ and $C$.
Specifically, we will choose (later) a parameter $\CProb{}$ close to 1/2.
As before, the set $A$ captures the points that are ``roughly equidistant''
between $x$ and $y$; specifically:
$A = \Set{i}{\Dist{i}{y} \le \Dist{i}{x} \le 1 + \Radius{A}} \subseteq Y$,
  where we will choose $\Radius{A}$ so that $\CProb{A} = \CProb{}$.
The set $B$ captures the points ``close to'' x:
$B = \Set{i}{\Dist{i}{x} \le \Radius{B}} \subseteq X$,
  where we choose $\Radius{B}$ so that $\CProb{B} = \CProb{}$.\footnote{
Notice that such $\Radius{A}, \Radius{B}$ exist without loss of
generality.
For if there were a radius $\rho$ such that --- say ---
$B^{\circ} = \Set{i}{\Dist{i}{x} < \rho}$ had
$\CProb{B^{\circ}} < \CProb{}$, while
$\overline{B} = \Set{i}{\Dist{i}{x} \leq \rho}$ had
$\CProb{\overline{B}} > \CProb{}$, 
we could split a point $i$ on the boundary
(that is, $i$ satisfies $\Dist{i}{x} = \rho$)
into two points, without affecting any outcomes for the instance.}
The set $C$ consists of the remaining points
  $C = \Compl{A \cup B}$.
($C$ may contain points from both $X$ and $Y$, and $\CProb{C} = 1-2\CProb{}$.)
Observe that the closer $\CProb{}$ is to $1/2$,
  the larger $\Radius{A}$ and $\Radius{B}$ will be,
  and we will have less control over where the points in $A$ and $B$
  are located.
Contrast this with the proof of Lemma~\ref{lem:given-half-point}, 
where the very stringent assumption of $\delta = 0$ allowed us to 
choose $\CProb{} = 1/2$ and still obtain $\Radius{A} = \Radius{B} = 0$.

We use the fact that $\Cost{x} \approx \half$ to derive that close to
half of the probability mass must be in $A$, and close to half in $B$.
To lower-bound the probability mass, notice that the cost of $x$ can
be lower-bounded term-by-term as follows:
\begin{enumerate}
\item All points in $B$ contribute cost at least 0.
\item All points in $X \setminus B$ contribute cost at least \Radius{B}.
\item All points in $A$ contribute cost at least 1.
\item All points in $Y \setminus A$ contribute cost at least $1+\Radius{A}$.
\end{enumerate}


We can thus lower bound \Cost{x} as follows:
\begin{align*}
\frac{1}{2-\delta} 
\geq \Cost{x} 
& \geq (\CProb{X} - \CProb{B}) \Radius{B} + \CProb{A} \cdot 1 
       + (\CProb{Y} - \CProb{A}) (1+\Radius{A})
\; = \; \CProb{X} \Radius{B} + \CProb{Y} (1+\Radius{A})
    - \CProb{B} \Radius{B} - \CProb{A} \Radius{A} \\
& \geq \half \Radius{B} + \half (1+\Radius{A}) 
    - \CProb{B} \Radius{B} - \CProb{A} \Radius{A}
\; = \; \half + (\half - \CProb{}) (\Radius{A} + \Radius{B}).
\end{align*}
This implies that $\Radius{A} + \Radius{B} \le \frac{\delta}{(2-\delta)(1-2\CProb{})}$.
In particular, notice that even while choosing 
the desired probability $\CProb{}$ very close to half
(e.g., $\CProb{} = 1/2 - O(\sqrt{\delta})$),
the cost bound still guarantees that such a \CProb{}
is achieved with small radii:
$\Radius{A} + \Radius{B} = O(\sqrt{\delta})$.
For notational convenience, we use $\Radius{} = \frac{\delta}{(2-\delta)(1-2\CProb{})}$
  to denote the upper bound for $\Radius{A} + \Radius{B}$,
  so $\Radius{B} \le \Radius{A} + \Radius{B} \le \Radius{}$.

\smallskip

The expected distortion can now be broken into terms based on the
three partitions $A, B, C$.
\begin{align*}
\Social{\DIST}{\CPROB}{\CPROB}
& = 2 \CProb{C} \left(\CProb{A}+\CProb{B}\right) 
           \Expect[i \sim \CPROB_C, j \sim \CPROB_{A \cup B}]{\Ratio{i}{j}} 
    + \CProb{C}^2 \cdot \Social{\DIST}{\CPROB_C}{\CPROB}
    + \CProb{B}^2 \cdot \Social{\DIST}{\CPROB_B}{\CPROB} \\
& \qquad + 2\CProb{A} \CProb{B} \Expect[i \sim \CPROB_A, j \sim \CPROB_B]{\Ratio{i}{j}} 
    + \CProb{A}^2 \cdot \Social{\DIST}{\CPROB_A}{\CPROB}
\end{align*}

We bound the terms in the sum separately.

\begin{itemize}
\item $\Expect[i \sim \CPROB_C, j \sim \CPROB_{A \cup B}]{\Ratio{i}{j}}
\leq 3-\delta$, simply because we assumed that the worst-case pairwise
distortion of any election was $3-\delta$.
\item The same bound of $3-\delta$ applies to
  \Social{\DIST}{\CPROB_C}{\CPROB}, for the same reason.


\item When both candidates $i,j$ are drawn from $B$,
assume that $i$ wins while $j$ has lower social cost.
We can use Inequalities~\eqref{eqn:triangle-winner} (for $i$ and $j$) and
\eqref{eqn:triangle-costs} (for $j$ and $x$, the latter having cost at
least $\half$) to bound
\begin{align*}
\frac{\Cost{i}}{\Cost{j}} 
& \le \frac{\Cost{j}+\frac{\Dist{i}{j}}{2}}{\Cost{j}}
\; \le \; 1 + \frac{\Radius{B}}{\Cost{j}} 
\; \le \; 1 + \frac{\Radius{B}}{1/2 - \Radius{B}}
\; = \; \frac{1}{1-2\Radius{B}}.
\end{align*}

We apply Lemma~\ref{lem:diff-with-cap}, and obtain that
\begin{equation*}
\Social{\DIST}{\CPROB_B}{\CPROB}
 \le \frac{1}{2} \left(1 + \frac{1}{1-2\Radius{B}} \right) 
\; = \; \frac{1-\Radius{B}}{1 - 2\Radius{B}}
\; \le \; \frac{1-\Radius{}}{1 - 2\Radius{}}.
\end{equation*}

\item The most difficult term to bound is
\begin{equation*}
2\CProb{A} \CProb{B} \Expect[i \sim \CPROB_A, j \sim \CPROB_B]{\Ratio{i}{j}} 
 + \CProb{A}^2 \cdot \Expect[i, j \sim \CPROB_A]{\Ratio{i}{j}}
\; = \;
\CProb{}^2 \cdot \Expect[i,j \sim \CPROB_A, b \sim \CPROB_B]{\Ratio{i}{b} +
\Ratio{j}{b} + \Ratio{i}{j}}.
\end{equation*}

Similar to the proof of Lemma~\ref{lem:given-half-point},
we define $\Delta_{i,j,b} = \Ratio{i}{b} + \Ratio{j}{b} + \Ratio{i}{j}$,
and upper-bound $\Delta_{i,j,b}$ for all $i, j \in A$ and $b \in B$ by
a quantity which tends to $5$ as $p \to 1/2$ and $\rho \to 0$.
The proof of the following lemma involves an intricate case analysis, and
is relegated to the end of this section. 
\begin{lemma}\label{lem:delta_ijb}
  $\Delta_{i,j,b} \leq 1 + 2 \cdot \RatioIwinsOverJ$ for all $i,j \in A$ and $b \in B$.
\end{lemma}
\end{itemize}

Substituting all the upper bounds into the expected distortion, 
we obtain that
\begin{align*}
\Social{\DIST}{\CPROB}{\CPROB}
& \le 2 \CProb{C} \cdot (3-\delta)
    + \CProb{B}^2 \cdot \frac{1-\Radius{}}{1 - 2\Radius{}}
    + \CProb{}^2 \cdot \max_{i,j,b} \Delta_{i,j,b} \\
& \leq 2 (1-2\CProb{}) \cdot (3-\delta)
    + \CProb{}^2 \cdot \frac{1-\Radius{}}{1 - 2\Radius{}}
    + \CProb{}^2 \left( 1 + 2 \cdot \RatioIwinsOverJ \right)
\end{align*}

Substituting $\CProb{} = \frac{1-\sqrt{\delta}}{2}$ (which may not be optimal),
we get $\Radius{} = \frac{\sqrt{\delta}}{2-\delta}$,
and a tedious manual calculation\footnote{or some help from Mathematica}
using the observation that
$2-\sqrt{\delta}-\delta = (1-\sqrt{\delta})(2+\sqrt{\delta})$ gives an upper
bound of
\begin{align*}
\Social{\DIST}{\CPROB}{\CPROB}
& \leq \frac{48 + 196 \delta^{0.5} - 348 \delta - 287 \delta^{1.5} + 275 \delta^{2} + 193 \delta^{2.5} - 39 \delta^{3} - 46 \delta^{3.5} - 8 \delta^{4}}{32 (1 + \delta^{0.5}/2)^2 (1 - \delta^{0.5}) (1 - \delta^{0.5} - \delta/2)}.
\end{align*}
Dropping dominated terms (negative in the numerator, positive in the
  denominator), this expression can be upper-bounded by
\begin{align*}
\Social{\DIST}{\CPROB}{\CPROB}
& \leq \frac{48 + 196 \sqrt{\delta}}{32(1+\sqrt{\delta}/2)^2(1-2\sqrt{\delta}+\delta/2)}.
\end{align*}

Finally, using the upper bound $\delta \leq \frac{1}{100}$,
we obtain that
\begin{align*}
\Social{\DIST}{\CPROB}{\CPROB}
& \leq (\frac{3}{2} + \frac{49}{8} \sqrt{\delta}) / (1-\frac{9}{8} \sqrt{\delta})
\; \leq \; 
(\frac{3}{2} + \frac{49}{8} \sqrt{\delta}) \cdot (1+\frac{90}{71} \sqrt{\delta})
\; \leq \; \frac{3}{2} + 9 \sqrt{\delta}.
\end{align*}

This completes the proof of Lemma~\ref{lem:two-minus-eps}.

\begin{extraproof}{Lemma~\ref{lem:delta_ijb}}
Note that, for all $i \in A$ and $b \in B$,
  we have
\[ 
1 - \Radius{B}
\; \le \; \Dist{i}{x} - \Dist{x}{b} 
\; \le \; \Dist{i}{b} 
\; \le \; \Dist{i}{x} + \Dist{x}{b} 
\; \le \; 1 + \Radius{A} + \Radius{B}.
\]

In the proof of Lemma~\ref{lem:given-half-point}, the three key
properties were that (1) $x$ was socially optimal, 
(2) any $i \in A$ was at most thrice worse than $x$, and
(3) if $i \in A$ beat $x$ in a pairwise election, then it was at most
twice worse than $x$.
The relaxed versions of these key properties are the following:

\begin{enumerate}
\item Every $b \in B$ is close to socially optimal: 
For all $i \in A$ and $b \in B$,

\begin{align*}
\frac{\Cost{b}}{\Cost{i}} 
& = \frac{\CProb{A}\Dist{b}{A}+\CProb{B}\Dist{b}{B}+\CProb{C}\Dist{b}{C}}{\CProb{A}\Dist{i}{A}+\CProb{B}\Dist{i}{B}+\CProb{C}\Dist{i}{C}}
\; \le \; \frac{(1-\CProb{B})\Dist{i}{b} + \CProb{B}\Dist{b}{B}}{\CProb{B}\Dist{i}{B}}\\
& \le \frac{(1-\CProb{B})(1+\Radius{A}+\Radius{B}) + \CProb{B} 2 \Radius{B}}{\CProb{B}(1-\Radius{B})}
\; \le \; \frac{(1-\CProb{})(1+\Radius{})+2\CProb{}\Radius{}}{\CProb{}(1-\Radius{})}
\; = \; \frac{1-\CProb{}+\Radius{}+\CProb{}\Radius{}}{\CProb{}(1-\Radius{})}.
\end{align*}

The first inequality is obtained by 
bounding $\Dist{b}{A} \leq \Dist{b}{i} + \Dist{i}{A}$
and $\Dist{b}{C} \leq \Dist{b}{i} + \Dist{i}{C}$,
then subtracting $\Dist{i}{A} \CProb{A} + \Dist{i}{C} \CProb{C}$ from
both the numerator and denominator.
The next step uses that 
$1 - \Radius{B} \le \Dist{i}{b} \le 1 + \Radius{A} + \Radius{B}$.
This ratio is at least 1 and approaches 1 as
$\Radius{} \to 0$ and $\CProb{} \to \half$.

\item For all $i \in A$ and $b \in B$ 
(regardless of who wins the pairwise election between them),
\[ 
\frac{\Cost{i}}{\Cost{b}} 
\; \le \; \frac{\Cost{b}+\Dist{i}{b}}{\Cost{b}} 
\; \le \; 1 + \frac{1+\Radius{A}+\Radius{B}}{\Cost{b}} 
\; \le \; 1 + \frac{1+\Radius{}}{\half(1-\Radius{})}
\; = \; \frac{3+\Radius{}}{1-\Radius{}}.
\]
This ratio is at least 3 and approaches 3 as $\Radius{} \to 0$.

\item For any $i \in A$ that wins the pairwise election against $b \in B$,
\[ 
\frac{\Cost{i}}{\Cost{b}} 
\; \le \; \frac{\Cost{b}+\frac{\Dist{i}{b}}{2}}{\Cost{b}} 
\; \le \; 1 + \frac{\frac{1}{2}(1+\Radius{A}+\Radius{B})}{\Cost{b}} 
\; \le \; 1 + \frac{\frac{1}{2}(1+\Radius{})}{\half(1-\Radius{})}
\; = \; \frac{2}{1-\Radius{}}.
\]
This ratio is at least 2 and approaches 2 as $\Radius{} \to 0$.
\end{enumerate}

We now fix $i, j \in A$ and $b \in B$, and upper-bound $\Delta_{i,j,b}$
through a detailed case analysis based on who wins (and is socially
better) in the three elections $(i,j), (i,b), (j,b)$.
Without loss of generality, we assume that $i$ wins the election
against $j$.
Throughout the analysis, as in the proof of
Lemma~\ref{lem:given-half-point},  
we use repeatedly that $\frac{t}{\Cost{b}} + \frac{\Cost{i}}{t}$
is a convex function of $t$, and in particular is maximized at
$t = \Cost{i}$ or $t = \Cost{b}$.

\begin{enumerate}
\item If the socially better candidate wins in at least two of the three
elections, then $\Delta_{i,j,b} \le 1 + 1 + 3-\delta = 5-\delta$,
because the third election can have distortion at most $3-\delta$.

\item If both $i$ and $j$ lose to $b$,
\[ 
\Delta_{i,j,b} 
\; \le \; 2 \left(\max_{i\in A, b\in B} \frac{\Cost{b}}{\Cost{i}} \right) 
+ 3 - \delta
\; \leq \; 
2 \cdot \RatioBoverI
+ 3 - \delta.
\]

\item If both $i$ and $j$ beat $b$, then we obtain
\[ 
\Delta_{i,j,b} 
\; = \; \max \{\frac{\Cost{i}}{\Cost{b}}, 1\} 
      + \max \{\frac{\Cost{j}}{\Cost{b}}, 1\} 
      + \max \{\frac{\Cost{i}}{\Cost{j}}, 1\}. 
\]
Because we are not in the first case, at most one of the three maxima
can be 1. There are three orderings of the social costs 
$\Cost{b}, \Cost{i}, \Cost{j}$ which are consistent with these
assumptions:
\begin{enumerate}
\item If $\Cost{b} \leq \Cost{i} \leq \Cost{j}$, then
\[
\Delta_{i,j,b} 
\; = \; \frac{\Cost{i}}{\Cost{b}} + \frac{\Cost{j}}{\Cost{b}} + 1
\; \leq \; 1 + 2 \cdot \RatioIwinsOverB.
\]
\item If $\Cost{b} \leq \Cost{j} \leq \Cost{i}$, then
\[
\Delta_{i,j,b} 
\; = \; \frac{\Cost{i}}{\Cost{b}} + \frac{\Cost{j}}{\Cost{b}} + 
\frac{\Cost{i}}{\Cost{j}}
\; \leq \; \frac{\Cost{i}}{\Cost{b}} + 1 + \frac{\Cost{i}}{\Cost{b}}
\; \leq \; 1 + 2 \cdot \RatioIwinsOverB.
\]
\item If $\Cost{j} \leq \Cost{b} \leq \Cost{i}$, then
\begin{align*}
\Delta_{i,j,b} 
& = \frac{\Cost{i}}{\Cost{b}} + 1 + \frac{\Cost{i}}{\Cost{j}}
\; \leq \; 1 + \left(\max_{i\in A, b \in B, i \text{ beats } b} 
                    \frac{\Cost{i}}{\Cost{b}} \right) 
      \left(1 + \max_{j\in A, b\in B} \frac{\Cost{b}}{\Cost{j}} \right) \\
& \leq 1 + \RatioIwinsOverB \cdot \left(1 + \RatioBoverI \right).
\end{align*}
\end{enumerate}
  
\item If $i$ beats $b$ and $j$ loses to $b$, then
\[ 
\Delta_{i,j,b} 
\; = \; \max \{\frac{\Cost{i}}{\Cost{b}}, 1\} 
      + \max \{\frac{\Cost{b}}{\Cost{j}}, 1\} 
      + \max \{\frac{\Cost{i}}{\Cost{j}}, 1\}. 
\]

Again, because at least two of the three pairwise elections result in
the socially worse candidate winning, we have only three cost
orderings consistent with the outcome:
\begin{enumerate}
\item If $\Cost{j} \le \Cost{i} \le \Cost{b}$, then
\[
\Delta_{i,j,b}
\; = \; 1 + \frac{\Cost{b}}{\Cost{j}} + \frac{\Cost{i}}{\Cost{j}}
\; \leq \; 
1 + \RatioBoverI + 3-\delta.
\]

\item If $\Cost{b} \le \Cost{j} \le \Cost{i}$, then
\[
\Delta_{i,j,b}
\; = \; \frac{\Cost{i}}{\Cost{b}} + 1 + \frac{\Cost{i}}{\Cost{j}}
\; \leq \;
1 + \RatioIwinsOverB \cdot \left(1 + \RatioBoverI \right),
\]
as in Case 3(c).

\item If $\Cost{j} \le \Cost{b} \le \Cost{i}$, then
\begin{align*} 
\Delta_{i,j,b} 
  & = \frac{\Cost{i}}{\Cost{b}} + \frac{\Cost{b}}{\Cost{j}} + \frac{\Cost{i}}{\Cost{j}}
\; \le \; 1 + 2 \cdot \frac{\Cost{i}}{\Cost{j}} \\
  & \le 1 + 2 \left( \max_{i \in A, b \in B, i \text{ beats } b} \frac{\Cost{i}}{\Cost{b}} 
                \cdot \max_{j \in A, b \in B} \frac{\Cost{b}}{\Cost{j}} \right)
\; \leq \; 1 + 2 \cdot \RatioIwinsOverJ,
\end{align*}
where the first inequality again used the convexity argument on
$\frac{\Cost{i}}{\Cost{b}} + \frac{\Cost{b}}{\Cost{j}}$.
\end{enumerate}

\item In the final case, $i$ loses to $b$ and $j$ beats $b$, resulting
  in a cycle in the election results. We now have
\[ 
\Delta_{i,j,b} 
\; = \; \max \{\frac{\Cost{b}}{\Cost{i}}, 1\} 
      + \max \{\frac{\Cost{j}}{\Cost{b}}, 1\} 
      + \max \{\frac{\Cost{i}}{\Cost{j}}, 1\}. 
\]
Again, we have three possible cost orderings consistent with the
assumption that at most one of the pairwise elections agrees with the
social costs:
\begin{enumerate}
\item If $\Cost{j} \le \Cost{i} \le \Cost{b}$, then
\[ 
\Delta_{i,j,b} 
\; = \; \frac{\Cost{b}}{\Cost{i}} + 1 + \frac{\Cost{i}}{\Cost{j}}
\; \leq \; 
1 + \RatioBoverI + 3-\delta.
\]
\item If $\Cost{i} \le \Cost{b} \le \Cost{j}$, then
\[ 
\Delta_{i,j,b} 
\; = \; \frac{\Cost{b}}{\Cost{i}} + \frac{\Cost{j}}{\Cost{b}} + 1
\; \leq \; 1 + \RatioBoverI + \RatioIwinsOverB.
\]
\item In the final case $\Cost{b} \le \Cost{j} \le \Cost{i}$, we again
  apply the convexity argument to bound
\[ 
\Delta_{i,j,b} 
\; = \; 1 + \frac{\Cost{j}}{\Cost{b}} + \frac{\Cost{i}}{\Cost{j}} 
\; \le \; 1 + 1 + \frac{\Cost{i}}{\Cost{b}} 
\; \le \; 2 + \RatioIoverB.
\]
\end{enumerate}
\end{enumerate}

Collecting all the upper bounds in all cases, we see that they are all
equal to (or immediately upper-bounded by) one of the following three
terms:
\[ \begin{cases}
2 \cdot \RatioBoverI + 3 - \delta
& \text{ for cases (1), (2), (4a), (5a)}\\
1 + 2 \cdot \RatioIwinsOverB
& \text{ for cases (3a), (3b), (5b), (5c)}\\
1 + 2 \cdot \RatioIwinsOverJ
& \text{ for cases (3c), (4b), (4c)}
\end{cases} \]

A somewhat tedious calculation shows that because
$\CProb{} \leq \half$, the expressions in the first and second cases
are always bounded by the expression in the third case.
This completes the proof of Lemma~\ref{lem:delta_ijb}.
\end{extraproof}


\section{Discussion and Open Questions}
\label{sec:conclusion}

We showed that under the simple model of two i.i.d.~candidates and a
majority election between them, government by the people is better for
the people if it is also of the people: 
there is a constant gap between the distortion caused by voting in the
case when $\VPROB = \CPROB$ vs. $\VPROB \neq \CPROB$.
For the case of the line, we pinned down the gap precisely, 
while for general metric spaces, we proved a small constant gap.

Our results can be construed as providing some mathematical
underpinnings for the benefits of \emph{lottocracy}.
\emph{Lottocracy} (also called \emph{sortition})
\cite{dowlen:sortition,guerrero:lottocracy,landemore:lottocracy}
refers to systems of government in which (some)
political officials are chosen through lotteries instead of
(or in addition to) elections.
Two of the arguments put forth in favor of lottocracy are:
(1) it is more inclusive \cite{landemore:lottocracy}, in the sense
that the office holders will be more representative of the population
as a whole and its different subgroups;
(2) it leads to more responsive government \cite{guerrero:lottocracy}:
because office holders are representative of the population, they will
respond more directly to the preferences of the population.
The definition of inclusiveness is very closely aligned with our
notion of candidates being ``of the people;''
it is sometimes justified by empirical and simulation studies giving
evidence that inclusive groups may be better at problem solving.
The notion of responsiveness is similar to our notion of government
being ``for the people;'' in this sense, our results could be --- with
some latitude --- rephrased as stating that inclusiveness may lead to
responsiveness. 

While most proponents of lottocracy argue in favor of filling offices
with randomly selected citizens, our analysis applies to a process
wherein voters do have a say, but the slate of candidates is random.
Allowing a vote between randomly selected candidates may in fact
address one of the main concerns about lottocracy, namely, the
competency of candidates \cite{guerrero:lottocracy,landemore:lottocracy}.
It simultaneously addresses a concern about democratic votes:
that the slate of candidates could be such that voters make a
societally suboptimal choice.
While the mathematical model presented here is far too simplistic to
provide reliable insights into the merits (or problems) of lottocracy
and its variants, 
it may serve as a point of departure for future more refined models.

In terms of more direct technical questions,
the most immediate open question is to obtain the maximum
expected distortion in general metric spaces.
We conjecture an upper bound of $3/2$.
Our conjecture is based on extensive computational experiments, and on
several partial results. In particular, we can show that the
distortion is upper-bounded by $3/2$ whenever the metric is uniform
(i.e., all voters/candidates are equidistant), or when there is a
location of the metric space that has half the
voters/candidates. Both properties seem to naturally arise in
worst-case constructions, although we are unable to prove at this
point that they are necessary for worst-case metrics.

Beyond the immediate open question, our work raises a number of other
directions for future work.
A first natural question is how the distortion depends on the metric
space. 
As we saw, the distortion for the line is $4-2\sqrt{2} < \frac{3}{2}$.
What is the distortion for $d$-dimensional Euclidean space?
Are there other natural metric spaces that are suitable models of
political or similar affiliation, and may be amenable to a detailed
analysis?

In this work, in order to isolate the issue of representativeness of
candidates, we focused on a majority election between two candidates.
When $k > 2$ candidates are running, vote aggregation becomes more
complex, and indeed, a large number of different voting rules have
been considered throughout history. 
The work of Anshelevich et
al.~\cite{anshelevich:bhardwaj:postl,anshelevich:postl:randomized}
analyzed the worst-case distortion of some of the most prevalent
voting rules. 
It would be interesting to examine the performance of these voting
rules under our model of candidates drawn from the voter population.
In particular, would such an analysis reveal a more fine-grained
stratification between some of the voting rules that perform equally
well (or poorly) under worst-case assumptions?

A further direction is to deviate from the extremes of worst-case
candidates or candidates drawn from the voter distribution.
How gracefully does the distortion degrade as the voter and candidate
distributions become more and more dissimilar?
Answering this question first requires a suitable definition of a
distance metric between probability distributions.
Such a definition will have to be ``Earthmover-like,'' yet also
``scale-invariant.''


\subsubsection*{Acknowledgments}
Yu Cheng is supported in part by Shang-Hua Teng's Simons Investigator Award. 
Shaddin Dughmi is supported by NSF CAREER award CCF-1350900 and NSF
Grant CCF-1423618. 
David Kempe is supported in part by NSF Grant CCF-1423618 and NSF
Grant IIS-1619458.
The authors would like to thank Elliot Anshelevich and Utkash Dubey
for useful conversations, and anonymous reviewers for helpful feedback.

\bibliographystyle{plain}
\bibliography{../bibliography/names,../bibliography/conferences,../bibliography/bibliography,../bibliography/voting}

\end{document}